\documentclass[12pt]{article}

\usepackage{color}
\usepackage[colorlinks=false]{hyperref}
\usepackage{amsmath, amssymb, amsthm}
\usepackage{mathtools}
\usepackage[noabbrev,capitalize,nameinlink]{cleveref}
\crefname{equation}{}{}
\usepackage{fullpage}
\usepackage{graphics}
\usepackage{pifont}
\usepackage{tikz}
\usepackage{bbm}
\usepackage[T1]{fontenc}

\DeclareMathOperator*{\argmin}{arg\,min}
\usetikzlibrary{arrows.meta}

\usepackage{environ}
\usepackage{framed}
\usepackage{url}
\usepackage[linesnumbered,ruled,vlined]{algorithm2e}
\usepackage[noend]{algpseudocode}
\usepackage[labelfont=bf]{caption}
\usepackage{cite}
\usepackage{framed}
\usepackage[framemethod=tikz]{mdframed}
\usepackage{appendix}
\usepackage{graphicx}
\usepackage{tcolorbox}
\usepackage[makeroom]{cancel}
\allowdisplaybreaks[1]

\usepackage{bm}
\usepackage{colonequals}
\usepackage[mathscr]{euscript}

\crefname{algocf}{Algorithm}{Algorithms}

\crefname{equation}{}{} 
\AtBeginEnvironment{appendices}{\crefalias{section}{appendix}} 

\crefname{algocf}{Algorithm}{Algorithms}

\numberwithin{equation}{section}
\newtheorem{theorem}{Theorem}[section]
\newtheorem{proposition}[theorem]{Proposition}
\newtheorem{lemma}[theorem]{Lemma}

\crefname{claim}{Claim}{Claims}

\newtheorem{corollary}[theorem]{Corollary}
\newtheorem{conjecture}[theorem]{Conjecture}
\newtheorem*{question*}{Question}

\newtheorem{definition}[theorem]{Definition}

\newtheorem{rem}[theorem]{Remark}

\newtheorem*{definition*}{Definition}

\theoremstyle{remark}

\newcommand{\one}{\bm{1}}

\newcommand{\EE}{\mathbb{E}}

\newcommand{\NN}{\mathbb{N}}
\newcommand{\PP}{\mathbb{P}}

\newcommand{\RR}{\mathbb{R}}
\renewcommand{\SS}{\mathbb{S}}

\DeclareSymbolFont{bbold}{U}{bbold}{m}{n}
\DeclareSymbolFontAlphabet{\mathbbold}{bbold}
\newcommand{\One}{\mathbbold{1}}

\newcommand{\ba}{\bm a}

\newcommand{\bg}{\bm g}
\newcommand{\bh}{\bm h}

\newcommand{\bu}{\bm u}
\newcommand{\bv}{\bm v}

\newcommand{\bx}{\bm x}
\newcommand{\by}{\bm y}

\newcommand{\bA}{\bm A}
\newcommand{\bB}{\bm B}

\newcommand{\bD}{\bm D}

\newcommand{\bG}{\bm G}

\newcommand{\bM}{\bm M}

\newcommand{\bP}{\bm P}

\newcommand{\bU}{\bm U}

\newcommand{\bW}{\bm W}
\newcommand{\bX}{\bm X}
\newcommand{\bY}{\bm Y}

\newcommand{\sN}{\mathcal{N}}

\newcommand{\sS}{\mathcal{S}}

\DeclareMathOperator{\GOE}{GOE}
\DeclareMathOperator{\sym}{sym}

\DeclareMathOperator{\Tr}{Tr}
\DeclareMathOperator{\rank}{rank}

\DeclareMathOperator{\Unif}{Unif}

\DeclareMathOperator{\Cov}{Cov}

\newcommand{\Ex}{\mathop{\mathbb{E}}}  
\newcommand{\Px}{\mathop{\mathbb{P}}}
\DeclareMathOperator{\symvec}{symvec}
\DeclareMathOperator{\symmat}{symmat}

\newcommand{\what}{\widehat}
\newcommand{\row}{\mathrm{row}}
\newcommand{\MACI}{\mathrm{MACI}}
\newcommand{\KACI}{\mathrm{KACI}}

\newcommand\numberthis{\addtocounter{equation}{1}\tag{\theequation}}

\newcommand{\Komlos}{Koml\'{o}s}
\newcommand{\Wish}{\mathrm{Wish}}

\allowdisplaybreaks

\title{Average-Case Matrix Discrepancy: Asymptotics and Online Algorithms}

\usepackage{authblk}

\author{Dmitriy Kunisky\thanks{Email: \texttt{dmitriy.kunisky@yale.edu}. Partially supported by ONR Award N00014-20-1-2335 and a Simons Investigator Award to Daniel Spielman.}\,}

\author{Peiyuan Zhang\thanks{Email: \texttt{peiyuan.zhang@yale.edu}. Partially supported by a Yale University Fund to Amin Karbasi.}\,}
\date{July 19, 2024}

\affil{Department of Computer Science, Yale University}

\begin{document}

\maketitle

\thispagestyle{empty}

\begin{abstract}
    We study the operator norm discrepancy of i.i.d.\ random matrices, initiating the matrix-valued analog of a long line of work on the $\ell^{\infty}$ norm discrepancy of i.i.d.\ random vectors.
    First, using repurposed results on vector discrepancy and new first moment method calculations, we give upper and lower bounds on the discrepancy of random matrices.
    We treat i.i.d.\ matrices drawn from the Gaussian orthogonal ensemble (GOE) and low-rank Gaussian Wishart distributions.
    In both cases, for what turns out to be the ``critical'' number of $\Theta(n^2)$ matrices of dimension $n \times n$, we identify the discrepancy up to constant factors.
    Second, we give a new analysis of the matrix hyperbolic cosine algorithm of Zouzias~(2011), a matrix version of an online vector discrepancy algorithm of Spencer~(1977) studied for average-case inputs by Bansal and Spencer~(2020), for the case of i.i.d.\ random matrix inputs.
    We both give a general analysis and extract concrete bounds on the discrepancy achieved by this algorithm for matrices with independent entries (including GOE matrices) and Gaussian Wishart matrices.
\end{abstract}

\clearpage

\tableofcontents

\pagestyle{empty}

\clearpage

\setcounter{page}{1}
\pagestyle{plain}

\section{Introduction}

We begin by reviewing the more familiar vector-valued versions of the matrix-valued questions we will study.
The discrepancy of a collection of input vectors $\bv_1, \dots, \bv_T \in \RR^n$, defined as
\begin{equation}
    \Delta(\bv_1, \dots, \bv_T) \colonequals \min_{\bx \in \{ \pm 1\}^T} \left\|\sum_{i = 1}^T x_i \bv_i \right\|_{\infty},
\end{equation}
is a classical quantity widely studied in combinatorics, combinatorial optimization, and computational geometry \cite{Spencer-1994-TenLecturesProbabilisticMethod,Matousek-1999-GeometricDiscrepancy,Chazelle-2000-DiscrepancyMethod}.
The following celebrated result of Spencer is one of the centerpieces of this literature.
\begin{theorem}[Spencer's ``six deviations suffice'' theorem \cite{Spencer-1985-SixDeviations,Spencer-1994-TenLecturesProbabilisticMethod}]
    \label{thm:spencer}
    If $\|\bv_i\|_{\infty} \leq 1$ for all $i \in [T]$, then $\Delta(\bv_1, \dots, \bv_T) \lesssim \sqrt{T \log(1 + \frac{n}{T})}$.
    Alternatively:
    \begin{equation}
        \Delta(\bv_1, \dots, \bv_T) \lesssim \left\{ \begin{array}{ll} \sqrt{T \log(n/T)} & \text{if } T \leq n/2, \\ \sqrt{n} & \text{if } T \geq n/2 \end{array} \right\}.
    \end{equation}
\end{theorem}

\begin{rem}[The case of many vectors]
    \label{rem:Tggn}
    Spencer's result is often stated not including the case of $T$ substantially larger than $n$, and indeed \cite{Spencer-1985-SixDeviations} only discussed $T = n$ and many subsequent works followed suit.
    Numerous works also state a suboptimal bound of the form $\sqrt{T\max\{1, \log(n / T)\}}$.
    The correct $O(\sqrt{n})$ bound when $T \geq n$ actually follows from the older ``iterated rounding'' method; see \cite{BF-1981-IntegerMakingTheorems,Barany-2008-LinearDependencies} as well as an unnumbered Corollary in Lecture~5 of \cite{Spencer-1994-TenLecturesProbabilisticMethod}.
\end{rem}
\noindent
What is interesting about Spencer's theorem is its improvement on the naive choice of i.i.d.\ random signs $x_i \sim \Unif(\{\pm 1\})$, which with high probability achieves only $\|\sum_{i =1}^T x_i \bv_i\|_{\infty} = \Theta(\sqrt{T \log n})$ for worst-case choices of $\bv_i$.
Once $T \gtrsim n$, Spencer's theorem says that the logarithmic factor in this bound may be removed with a more careful choice of $x_i$; when $T \gg n$, random signs are severely suboptimal and the aforementioned iterated rounding reduction must be applied to obtain an accurate bound.

A natural follow-up question is whether efficient algorithms can find such assignments of the signs $x_i$.
Several such algorithms were eventually discovered to achieve the same discrepancy as Spencer's non-constructive proof \cite{Bansal-2010-ConstructiveDiscrepancyMinimization,LM-2015-ConstructiveDiscrepancyRandomWalk,Rothvoss-2017-ConstructiveDiscrepancyMinimization} (see also numerous related algorithmic ideas in \cite{DGLN-2019-TowardsConstructiveBanaszczyk,BDGL-2018-GramSchmidtWalk,BDG-2019-AlgorithmKomlosBanaszczyk,HSSZ-2019-BalancingCovariatesGramSchmidt,ALS-2021-DiscrepancySelfBalancingWalk,LSS-2021-GaussianFixedPointWalk,BLV-2022-UnifiedDiscrepancyMinimization}).

Beyond that, one may also ask the more stringent question of whether \emph{online} algorithms---ones where $x_i$ can only depend on $(\bv_1, \dots, \bv_i)$---can achieve the same discrepancy.
In this case, in the \emph{adversarial} setting of Spencer's theorem, the answer is no: if $\bv_i$ may be chosen by an adversary depending on $x_1, \dots, x_{i - 1}$, then the adversary may force any online algorithm to produce $x_i$ with $\|\sum_{i =1}^T x_i \bv_i\|_{\infty} \gtrsim \sqrt{T \log n}$ \cite[Lecture 4]{Spencer-1994-TenLecturesProbabilisticMethod}.
That is, in the adversarial online setting, one cannot improve on choosing the signs $x_i$ at random (in contrast to Spencer's theorem on the offline setting).

However, it remains interesting to analyze online algorithms if we shift our attention to the \emph{average-case} setting, where the $\bv_i$ are chosen randomly rather than adversarially.\footnote{Another setting that may be viewed as lying in between these two is the \emph{oblivious adversarial} setting, where the adversary must fix $\bv_1, \dots, \bv_T$ in advance, but we focus on the average case here.}
In this case, the typical true discrepancy is understood quite precisely for different scalings of $T$ and choices of $\bv_i$.
The following result, while its most general part does not quite fall in Spencer's setting of $\|\bv_i\|_{\infty} \leq 1$, accurately captures the main phenomenon (see \cite{TMR-2020-BalancingGaussianVectors} or \cite{GKPX-2023-OGPDiscrepancy} for some discussion of other distributions).

\begin{theorem}[Theorem 1 and Remark 2 of \cite{TMR-2020-BalancingGaussianVectors}]
    \label{thm:gaussian-vectors}
    Suppose $T = T(n)$ satisfies $T \geq \epsilon n$ for some $\epsilon > 0$.
    If $\bv_1, \dots, \bv_T \sim \sN(\bm 0, \bm I_n)$ independently, then with high probability as $n \to \infty$,
    \begin{equation}
        \Delta(\bv_1, \dots, \bv_T) \asymp_{\epsilon} \sqrt{T} \, 2^{-T/n}.
    \end{equation}
    The same holds for $\bv_i \sim \Unif([-1, 1]^n)$ so long as $n \leq T / \log T$.
\end{theorem}
\noindent
In words, if we think of the typical average-case discrepancy as a function of the number $T$ of vectors to be balanced, the discrepancy increases to a maximum order of $O(\sqrt{n})$ as $T$ increases to order $\Theta(n)$, just like the worst-case discrepancy per Spencer's theorem.
However, unlike the worst-case discrepancy, as $T$ grows beyond $\Theta(n)$, the discrepancy then \emph{decreases} because it becomes possible to achieve greater and greater cancellations with more and more vectors to balance.

In this setting, the following results show that online algorithms can achieve optimal or nearly optimal discrepancy for small $T = \Theta(n)$; however, these algorithms do not capture the decaying discrepancy once $T \gg n$.\footnote{One intuition for this is that the cancellations possible in a long stream of vectors involve vectors far apart in the stream, but an online algorithm cannot ``backtrack'' and change the signs of vectors it has seen already to achieve these cancellations.}
\begin{theorem}[Section 3 of \cite{BS-2020-OnlineBalancingRandom}, Section 2.1 of \cite{BJSS-2020-OnlineVectorBalancingDiscrepancy}]
    \label{thm:vector-online}
    The following hold:
    \begin{enumerate}
    \item If $\bv_1, \dots, \bv_T \sim \Unif(\{\pm 1\}^n)$ independently, then there is an online algorithm that, when run on these vectors, with high probability outputs a signing $x_1, \dots, x_T \in \{\pm 1\}$ with $\|\sum_{i = 1}^T x_i\bv_i \|_{\infty} = O(\sqrt{n})$.
    \item If $\bv_1, \dots, \bv_T \sim \mu^{\otimes n}$ (i.e., all having i.i.d.\ entries drawn from $\mu$) independently where $\mu$ is a probability measure supported on $[-1, 1]$ and having mean zero, then there is an online algorithm that, when run on these vectors, with high probability outputs a signing $x_1, \dots, x_T \in \{\pm 1\}$ with $\|\sum_{i = 1}^T x_i\bv_i \|_{\infty} = O(\sqrt{n} \log T)$.
    \end{enumerate}
\end{theorem}
\noindent
We note that we refer here to a side discussion of \cite{BJSS-2020-OnlineVectorBalancingDiscrepancy}; their main result gives a weaker bound but allows for correlations in the entries of the $\bv_i$.
See Section~\ref{sec:related} for more discussion of the rather delicate dependence of these results on the distribution of the $\bv_i$.
See also the recent results of \cite{GKPX-2023-OGPDiscrepancy} for a demonstration that the performance of the algorithm in Part 1 of Theorem~\ref{thm:vector-online} above is optimal up to constants when $T = \Theta(n)$.

In this paper, we initiate the study of the questions discussed above for the related problem of \emph{matrix discrepancy}.
Let us review the parallel definitions and background.
For a collection of matrices $\bA_1, \dots, \bA_T \in \RR^{n \times n}_{\sym}$, we set (harmlessly overloading the $\Delta(\cdot)$ notation)
\begin{equation}
    \Delta(\bA_1, \dots, \bA_T) \colonequals \min_{\bx \in \{\pm 1\}^T} \left\|\sum_{i = 1}^T x_i \bA_i \right\|,
\end{equation}
where the norm is the operator norm.

In this case, the analog of Spencer's theorem (Theorem~\ref{thm:spencer}) remains only conjectural.
\begin{conjecture}[Matrix Spencer conjecture \cite{Meka-2014-DiscrepancyMatrixSpencer,Zouzias-2011-MatrixHyperbolicCosine}]
    \label{conj:mx-spencer}
    For all $\bA_1, \dots, \bA_T \in \RR^{n \times n}_{\sym}$ with $\|\bA_i\| \leq 1$,
    \begin{equation}
        \Delta(\bA_1, \dots, \bA_T) \lesssim \left\{ \begin{array}{ll} \sqrt{T \log(n/T)} & \text{if } T \leq n/2, \\ \sqrt{T} & \text{if } n/2 \leq T \leq n^2, \\ n & \text{if } T \geq n^2 \end{array} \right\}.
    \end{equation}
\end{conjecture}
\noindent
As with Spencer's theorem, the various regimes of $T$ are somewhat delicate to handle correctly.
In the matrix setting, the iterated rounding argument cited in Remark~\ref{rem:Tggn} only applies to $T \gtrsim n^2$ (see the Introduction of \cite{DJR-2021-MatrixDiscrepancyMirrorDescent} for discussion of this point), and the first two regimes above give the main substance of the conjecture.

Note that, if $T = O(n)$, then Conjecture~\ref{conj:mx-spencer} contains Theorem~\ref{thm:spencer} as the special case where all $\bA_i$ commute and therefore are simultaneously diagonalizable.
In contrast, the non-commutative Khintchine inequality\footnote{A matrix concentration inequality, not to be confused with the other Khintchine inequalities discussed later in this paper, which are vector anti-concentration inequalities.} of Pisier and Lust-Piquard \cite{LPP-1991-NonCommutativeKhintchinePaley} only implies that choosing random signs $x_i \sim \Unif(\{\pm 1\})$ achieves
    \begin{equation}
        \label{eq:random-signs}
        \EE \left\|\sum_{s = 1}^T x_s \bA_s \right\|
        \lesssim \left\| \sum_{s = 1}^T \bA_s^2 \right\|^{1/2} \sqrt{\log n} \\
        \leq \sqrt{T\log n},
    \end{equation}
just as in the vector case.

Despite remarkable recent progress on special cases \cite{HRS-2021-MatrixDiscrepancyQuantumCommunication,DJR-2021-MatrixDiscrepancyMirrorDescent,BJM-2022-MatrixSpencerSparse}, the full version of Conjecture~\ref{conj:mx-spencer} appears difficult to attack and is perhaps one of the main open problems in discrepancy theory at the time of writing.
The matrix versions of the remaining results discussed above in either online or average-case settings, on the other hand, are more approachable.
Towards both probing Conjecture~\ref{conj:mx-spencer} and exploring the parallel problem of the discrepancy of random rather than worst-case $\bA_i$, our goal in this paper will be to demonstrate matrix-valued analogs of Theorems~\ref{thm:gaussian-vectors} and \ref{thm:vector-online}.

\subsection{Notation}
We use bold small and capital letters to denote vectors and matrices, respectively, and non-bold letters with subscripts to denote entries: $\ba \in \RR^n; \bA \in \RR^{n\times n}; a_i, A_{ij} \in \RR$.

For a matrix $\bA \in \RR^{n \times n}_{\sym}$, we write $\lambda_1(\bA) \geq \cdots \geq \lambda_n(\bA)$ for its ordered eigenvalues.
We denote the operator norm by $\|\bA\|$, the trace or nuclear norm by $\|\bA\|_{*} \colonequals \Tr(\sqrt{\bA\bA^{\top}})$, and the Frobenius norm by $\|\bA\|_{F} \colonequals \sqrt{\Tr(\bA\bA^{\top})}$.
For an additional $\bB \in \RR^{n \times n}_{\sym}$, we write $\langle \bA, \bB \rangle \colonequals \Tr(\bA \bB) = \sum_{i, j = 1}^n A_{ij}B_{ij}$.
We write $\row(\bA)$ for the row space (equivalently column space or image) of $\bA$.
We write $\symvec(\bA) \in \RR^{n(n + 1)/2}$ for the vector with entries $A_{ii}$ for $1 \leq i \leq n$ and $\sqrt{2}A_{ij}$ for $1 \leq i < j \leq n$, so that $\langle \bA, \bB \rangle = \langle \symvec(\bA), \symvec(\bB) \rangle$.
We write $\symmat(\ba) \in \RR^{n \times n}_{\sym}$ for $\ba \in \RR^{n(n + 1)/2}$ for the inverse transformation.
For a subspace $V \subset \RR^n$, we write $\bP_V$ for the orthogonal projection matrix to $V$.

The asymptotic notations $O(\cdot), o(\cdot), \Omega(\cdot), \omega(\cdot), \Theta(\cdot)$, $\lesssim$, $\gtrsim$, $\ll$, $\gg$ have their standard meanings, always in the limit $n \to \infty$.
Subscripts of these symbols indicate quantities that the implicit constants depend on.

\subsection{Main Results: Discrepancy Asymptotics}

We first consider asymptotic results on average-case matrix discrepancy.
We note that here we are discussing bounds on the ``true'' discrepancy value, not lower or upper bounds for particular classes of algorithms.
Such results are still useful for evaluating algorithms; we in part view our asymptotic results as justifying that our analysis of an online algorithm below is close to optimal.

Before giving concrete results, we state the following informal conjecture on the kind of lower bound that we expect for ``sufficiently nice'' sequences of i.i.d.\ random matrices.
We put in quotes the terms in the statement to which we are not giving a precise definition.

\begin{conjecture}[Informal discrepancy asymptotics]
    \label{conj:lb}
    Suppose that $\bA_1, \dots, \bA_T \in \RR^{n \times n}_{\sym}$ are i.i.d.\ random matrices for $T = T(n)$ with $\|\bA_i\| \leq 1$ almost surely, having a ``sufficiently nice'' distribution, and having ``effective rank'' $r = r(n)$.
    Then, with high probability,
    \begin{equation}
        \Delta(\bA_1, \dots, \bA_T) \asymp \sqrt{\frac{rT}{n}} 4^{-T / n^2}.
    \end{equation}
\end{conjecture}
\noindent
Here, we may interpret $4^{-T/n^2} = 2^{-T/(n^2/2)}$, where $n^2/2$ is approximately the dimension of $\RR^{n \times n}_{\sym}$.
With this interpretation, up to the factor of $\sqrt{r / n}$ accounting for the possibility of low-rank matrices, this prediction is identical to the result of \cite{TMR-2020-BalancingGaussianVectors} stated in our Theorem~\ref{thm:gaussian-vectors} for vector discrepancy.
See Section~\ref{sec:lb} for a heuristic justification of this conjecture through a first moment calculation.

For our rigorous results, we first consider the special case of i.i.d.\ matrices drawn from the \emph{Gaussian orthogonal ensemble (GOE)}, for which we can give a detailed analysis matching the above.

\begin{definition}[Gaussian orthogonal ensemble]
    \label{def:goe}
    We write $\GOE(n, \sigma^2)$ for the law of the $n \times n$ random symmetric matrix $\bW$ with independent entries on and above the diagonal distributed as $W_{ij} = W_{ji} \sim \sN(0, (1 + \One\{i = j\})\sigma^2)$.
\end{definition}
\noindent
For the sake of comparison with Conjecture~\ref{conj:lb}, we mention the following standard result of random matrix theory on the typical norm of GOE matrices.
\begin{proposition}[GOE norm]
    For a sequence $\sigma = \sigma(n) > 0$, if $\bW \sim \GOE(n, \sigma^2)$, then with high probability $\|\bW\| = (1 + o(1)) \cdot 2\sigma \sqrt{n}$.
\end{proposition}

\begin{theorem}[Asymptotics: GOE matrices]
    \label{thm:goe}
    Suppose $\bA_1, \dots, \bA_T \sim \GOE(n, \frac{1}{n})$ are i.i.d.\ for an increasing sequence $T = T(n)$ .
    Then, for absolute constants $0 < c < C$, with high probability as $n \to \infty$,
    \begin{equation}
        c \sqrt{T} 4^{-T/n^2} \leq \Delta(\bA_1, \dots, \bA_T) \leq Cn.
    \end{equation}
\end{theorem}
\noindent
When $T \sim n^2$, the bounds match up to constants and give $\Delta(\bA_1, \dots, \bA_T) \asymp n$ with high probability, so in this case our result correctly identifies the scaling of the typical discrepancy.
We think of this as a ``critical'' regime where the discrepancy is maximized, a natural analog to the regime $T \sim n$ for the vector case much studied as the ``symmetric binary perceptron'' in statistical physics (see Section~\ref{sec:related}).
This also shows that the matrix Spencer conjecture is tight in this regime even for the simple and natural example of i.i.d.\ random matrices.

We also give a general-purpose result that can be applied to broader classes of distributions in the same regime $T \sim n^2$.
\begin{theorem}[Asymptotics: general distributions]
    \label{thm:lb}
    Let $\mu_n$ be a probability measure on $\RR^{n \times n}_{\sym}$ having mean zero, and let $r = r(n) \in \NN$ have $r(n) \leq n$.
    Suppose that the following hold, where we let $\bA \sim \mu_n$:
    \begin{enumerate}
    \item $\|\bA\|_F^2 \in [\frac{1}{2}r, \frac{3}{2}r]$ with probability at least $1 - o(\frac{1}{n})$.
    \item There is a constant $c > 0$ not depending on $n$ such that, for any $\bY \in \RR^{n \times n}_{\sym}$ with $\|\bY\|_F = 1$, $\EE \exp(c \frac{n}{\sqrt{r}} |\langle \bA, \bY \rangle|) \leq 2$.
    \end{enumerate}
    Then, there exists an $\epsilon_0 > 0$ depending only on $c$ above such that, if $T = T(n) = \epsilon n^2$ with $\epsilon < \epsilon_0$ and  $\bA_1, \dots, \bA_T \sim \mu_n$ are i.i.d., then we have with high probability as $n \to \infty$
    \begin{equation}
        \Delta(\bA_1, \dots, \bA_T) \asymp_{\epsilon} \sqrt{rn}.
    \end{equation}
\end{theorem}
\noindent
The parameter $r$ should be thought of as an ``effective rank'' of $\bA \sim \mu_n$, which is compatible with Condition 1 above when $\|\bA\| = O(1)$ with high probability.

The following application uses this to treat a convenient random matrix distribution where the rank is a free parameter (unlike for the GOE), for which the conditions above are not difficult to verify.\footnote{Actually, Theorem~\ref{thm:lb} cannot quite be used in a black-box manner because the distribution involved here is not centered, but we will use the same proof \emph{mutatis mutandis}.}

\begin{definition}[Wishart ensemble]
    Let $r, n \in \NN$, $\sigma^2 > 0$, and $\bG \in \RR^{n \times r}$ be a random matrix with i.i.d.\ entries distributed as $\sN(0, \sigma^2)$.
    We write $\Wish(n, r, \sigma^2)$ for the law of $\bG\bG^{\top} \in \RR^{n \times n}_{\sym}$.
    Note that $\rank(\bW) = \min\{r, n\}$ when $\bW \sim \Wish(n, r, \sigma^2)$.
\end{definition}

\noindent
As for the GOE, it will be useful to register the following result about the norm of a Wishart matrix.
See, e.g., Theorem 7.3.1 and Corollary 7.3.3 of \cite{Vershynin-2018-HDP} for proofs.

\begin{proposition}[Wishart norm]
    \label{prop:wish-norm}
    For sequences $r = r(n) \in \NN$ and $\sigma = \sigma(n) > 0$, if $\bW \sim \Wish(n, r, \sigma^2)$, we have with high probability $\|\bW\| = (1 + o(1)) \sigma^2(\sqrt{n} + \sqrt{r})^2$.
\end{proposition}

\begin{corollary}[Asymptotics: Wishart matrices]
    \label{cor:lb-proj}
    Let $r: \NN \to \NN$ be such that $r(n) \leq n$ for all $n$ and $r(n) = o(n)$ as $n \to \infty$.
    Let $\mu_n = \Wish(n, r, \frac{1}{n})$.
    (Note that, by Proposition~\ref{prop:wish-norm}, this means that $\bW \sim \mu_n$ has $\|\bW\| = O(1)$ with high probability.)
    If $\bA_1, \dots, \bA_T \sim \mu_n$ are i.i.d.\ and $T = T(n) = \epsilon n^2$ for sufficiently small $\epsilon$, then, with high probability,
    \begin{equation}
        \Delta(\bA_1, \dots, \bA_T) \asymp_{\epsilon} \sqrt{rn}.
    \end{equation}
\end{corollary}

\subsection{Main Results: Matrix Hyperbolic Cosine Algorithm}

We next study an online algorithm for matrix discrepancy, which is a matrix-valued version of the algorithm achieving Part 2 of Theorem~\ref{thm:vector-online}.
The latter algorithm, which minimizes a potential function involving the hyperbolic cosine function, was introduced by \cite{Spencer-1977-BalancingGames} and studied in the average-case setting by \cite{BS-2020-OnlineBalancingRandom} (see also \cite[Lecture 4]{Spencer-1994-TenLecturesProbabilisticMethod}).

The simple algorithm of choosing random signs $x_1, \dots, x_T \sim \Unif(\{\pm 1\})$ is online, and we have seen in \eqref{eq:random-signs} that it achieves a discrepancy of $O(\sqrt{T\log n})$.
On the other hand, say for $\bA_i \sim \GOE(n, \frac{1}{n})$, we have seen that the true discrepancy is $O(n)$ even for $T \gg n^2$.
Our goal below will be to analyze a more sophisticated algorithm that matches this behavior; our analysis focuses on this regime, and will give inferior results for $T \ll n^2$.

The \emph{matrix hyperbolic cosine (MHC) algorithm} we study is given in Algorithm~\ref{alg:mhc}.
The only previous work we are aware of that has studied this algorithm before is \cite{Zouzias-2011-MatrixHyperbolicCosine}, whose results applied to our average-case setting are weaker than ours; see Section~\ref{sec:related} for details.
In short, this algorithm maintains a single $n \times n$ matrix $\bM$ of the current signed sum of the matrices it has been given, and, given a new matrix, chooses a sign to minimize the potential function $\Tr \cosh(\alpha \bM)$ for a suitable constant $\alpha > 0$.\footnote{Recall the basic definitions of transcendental matrix functions: $\exp(\bX) \colonequals \sum_{k \geq 0} \frac{1}{k!}\bX^k$, and $\cosh(\bX) \colonequals \frac{1}{2}(\exp(\bX) + \exp(-\bX))$. See Section~\ref{sec:linalg} for further details.}

\begin{algorithm}
\SetKwInOut{Input}{Input\,}
\SetKwInOut{Output}{Output\,}
\caption{Matrix hyperbolic cosine (MHC) algorithm for online discrepancy.}\label{alg:mhc}
\Input{Stream of $\bA_1, \bA_2, \ldots \in \RR^{n \times n}_{\sym}$ and $\alpha > 0$.}
\Output{Stream of $x_1, x_2, \ldots \in \{ \pm 1 \}$.} 
    $\bM \gets \bm 0$\;
    \For{$i \geq 1$}{
        \textbf{read} $\bA_i$\;
        $x_i \gets \argmin_{x \in \{\pm 1\}} \Tr \cosh(\alpha(\bM + x\bA_i))$\;
        $\bM \gets \bM + x_i \bA_i$\; 
        \textbf{yield} $x_i$\;
    }
\end{algorithm}

We next introduce two conditions on a random matrix distribution that together will suffice to establish upper bounds on the discrepancy that the MHC algorithm achieves.
The first is an \emph{anti-concentration inequality}, analogous to standard such conditions for random vectors (e.g., \cite{PS-1995-KhintchineInequalities,Krishnapur-2016-AntiConcentration}).

\begin{definition}[Matrix anti-concentration]
    Let $\mu$ be a probability measure on $\RR^{n \times n}_{\sym}$.
    We say that $\mu$ satisfies a \emph{matrix anti-concentration inequality} with \emph{constant} $\eta > 0$ if, for all $\bX \in \RR^{n \times n}_{\sym}$, we have
    \begin{equation}
        \Ex_{\bA \sim \mu}|\langle \bX, \bA \rangle| \geq \frac{\eta}{\sqrt{n}} \|\bX\|_*.
    \end{equation}
    When this is the case, we write that $\mu$ is $\MACI(\eta)$.
\end{definition}
\noindent
    We will only ever prove a MACI condition through a stronger Khintchine-like inequality (see Definition~\ref{def:kaci}).
    We state our main result in terms of this weaker condition to draw a parallel with \cite{BJSS-2020-OnlineVectorBalancingDiscrepancy} where such a condition was used in the vector case to prove weaker discrepancy results over a broader range of distributions.

The second condition is a \emph{quantitative isotropy} condition for the random row space of our random matrices.

\begin{definition}[Unbiasedness]
    Let $\mu_n$ be a probability measure on $\RR^{n \times n}_{\sym}$ for each $n \geq 1$.
    We say that this sequence is \emph{$\theta$-unbiased} with \emph{rank sequence} $r: \NN \to \NN$ if, for all $n \geq 1$,
    \begin{equation}
        \left\|\Ex_{\bA \sim \mu_n}\bP_{\row(\bA)}\right\| \leq \theta \frac{r(n)}{n}.
    \end{equation}
\end{definition}

\noindent
For intuition, first note that the condition is only non-trivial when $r \ll n$, so that $\bA$ is low-rank.
If $\row(\bA)$ is a uniformly random subspace of $\RR^n$ having dimension $r$, then by symmetry we have $\EE \bP_{\row(\bA)} = \frac{r}{n}\bm I_n$.
The unbiasedness condition therefore says that $\row(\bA)$ is quantitatively close to uniformly random.

These two conditions imply the following upper bound on the discrepancy achieved by the signs that Algorithm~\ref{alg:mhc} outputs.

\begin{theorem}[General analysis of Algorithm~\ref{alg:mhc}]
    \label{thm:mhc}
    Suppose $\mu_n$ is a probability measure on $\RR^{n \times n}_{\sym}$ for each $n \geq 1$ and $r: \NN \to \NN$.
    Suppose that:
    \begin{enumerate}
        \item $\mu_n$ is supported on matrices of operator norm at most 1 for each $n$,
        \item the sequence $\mu_n$ is $\theta$-unbiased with rank sequence $r(n)$ for a constant $\theta > 0$, and
        \item the law of $\frac{n}{\sqrt{r(n)}}\bA$ for $\bA \sim \mu_n$ is $\MACI(\eta)$ for each $n$ and a constant $\eta > 0$.
    \end{enumerate}
    Then, Algorithm~\ref{alg:mhc} run with $\alpha = \frac{1}{4\sqrt{2}} \frac{\eta}{\theta} \frac{1}{\sqrt{rn}}$ and a stream of $T$ i.i.d.\ random matrices drawn from $\mu_n$ outputs a stream of signs $x_1, \dots, x_T \in \{ \pm 1\}$ that achieves, with high probability as $n \to \infty$,
    \begin{equation}
        \max_{t = 1}^T\left\|\sum_{s = 1}^t x_s \bA_s \right\| \lesssim_{\eta, \theta} \sqrt{rn} \log T.
    \end{equation}
\end{theorem}
\noindent
We note that the logarithmic term in this bound is the same as in the general version of the vector bound of \cite{BJSS-2020-OnlineVectorBalancingDiscrepancy}, as stated in our Theorem~\ref{thm:vector-online}.
Also, the analysis of \cite{Zouzias-2011-MatrixHyperbolicCosine} would only give a bound of $\sqrt{T \log n}$, which is much larger than the above for large $T$ and indeed is just the same as the non-commutative Khintchine inequality implies for random signs; cf.\ the inequality stated in \eqref{eq:random-signs}.
Zouzias' purpose in analyzing the MHC algorithm was precisely to achieve the performance of random signs with a deterministic algorithm, while we are interested here in a sharper characterization of the algorithm's performance.

\begin{rem}[Comparison with vectorized algorithms]
    \label{rem:als}
    Another reasonable approach to solving matrix discrepancy problems is to vectorize the matrices and use vector discrepancy algorithms.
    This is especially powerful since some vector discrepancy algorithms guarantee that their output is a \emph{subgaussian} vector, and the subgaussianity of a vectorized matrix immediately implies a high-probability norm bound, as we derive in Proposition~\ref{prop:subgauss-mx} and use in our non-algorithmic upper bounds.
    To the best of our knowledge, the best such bound in the literature for an online algorithm is due to \cite{ALS-2021-DiscrepancySelfBalancingWalk}, whose bound in our setting would replace $\log(T)$ with $\log(nT)$, which is inferior in some regimes (see also our Remark~\ref{rem:mhc-improvements} for how our $\log(T)$ term can be improved to accurately treat the case $T \ll n$).
    More conceptually, we believe that our algorithm, being similar to that of \cite{BS-2020-OnlineBalancingRandom}, is more likely to be adaptable to sometimes remove this logarithmic factor altogether (which \cite{BS-2020-OnlineBalancingRandom} achieve, as stated in our Theorem~\ref{thm:vector-online}, for the special distribution of vectors drawn uniformly from $\{\pm 1\}^n$).
\end{rem}

We now give several applications of Theorem~\ref{thm:mhc} to different matrix distributions.
The proofs of these, given in Section~\ref{sec:mhc-proofs}, amount to verifying the unbiasedness and MACI properties for the underlying distributions.

\begin{corollary}[Independent entries]
    \label{cor:mhc-ind}
    For each $n \geq 1$, let $\widetilde{\mu}_n$ be a probability measure on $\RR^{n \times n}_{\sym}$ so that $\bA \sim \widetilde{\mu}_n$ has:
    \begin{enumerate}
    \item The entries of $\bA$ on and above the diagonal are independent,
    \item $\EE A_{ij} = 0$ for all $i, j \in [n]$,
    \item $C_1^2 / n \leq \EE A_{ij}^{2} \leq C_2^2 / n$ for all $n \geq 1$ and $i, j \in [n]$ for some $C_1, C_2 > 0$,
    \item $\EE A_{ij}^{4} / (\EE A_{ij}^{2})^2 \leq C_3$ for all $n \geq 1$ and $i, j \in [n]$ for some $C_3 > 0$, and
    \item $\|\bA\| \leq 1$ with high probability as $n \to \infty$.
    \end{enumerate}
    Let $\mu_n$ be $\widetilde{\mu}_n$ conditioned on the event that $\|\bA\| \leq 1$.
    Then, the conclusion of Theorem~\ref{thm:mhc} applies to these $\mu_n$ with $r(n) \colonequals n$.
    That is, Algorithm~\ref{alg:mhc} run with $\alpha = \gamma / n$ for a suitable constant $\gamma = \gamma(C_1, C_2, C_3) > 0$ and a stream $\bA_1, \dots, \bA_T$ drawn i.i.d.\ from $\mu_n$ outputs a stream of signs $x_1, \dots, x_T \in \{ \pm 1 \}$ that achieves, with high probability as $n \to \infty$,
    \begin{equation}
        \max_{t = 1}^T\left\|\sum_{s = 1}^t x_s \bA_s \right\| \, \lesssim_{C_1, C_2, C_3} \, n \log T.
    \end{equation}
\end{corollary}

The statement in Corollary~\ref{cor:mhc-ind} applies in particular to any reasonable model of a Wigner random matrix, such as the GOE, symmetric matrices with i.i.d.\ Rademacher entries, and so forth, so long as these matrices are rescaled to make the typical operator norm less than~1.
Moreover, we allow for entries with different distributions so long as they satisfy the regularity conditions uniformly, so the statement also applies to, e.g., symmetric Gaussian random matrices with different variance profiles (see \cite{BvH-2016-RandomMatrixIndependentEntries,vH-2017-StructuredRandomMatrices} for some discussion of the challenges of working with such matrices beyond classical models with i.i.d.\ entries).
Comparing with Theorem~\ref{thm:goe} above reveals a similar phenomenon to the vector case: up to logarithmic terms, online algorithms admit an average-case bound matching the true average-case discrepancy for $T = \Theta(n^2)$, the scaling of $T$ for which the true discrepancy is largest.

We also give the following application, where $\bA_1, \dots, \bA_T$ are a normalization of Gaussian Wishart matrices and the rank of $\bA_i$ may be any sequence, illustrating the non-trivial rank dependence in Theorem~\ref{thm:mhc}.

\begin{corollary}[Normalized Wishart matrices]
    \label{cor:wishart}
    Let $r: \NN \to \NN$ be a non-decreasing function.
    We abbreviate $r \colonequals r(n)$.
    For each $n \in \NN$, let $\mu_n$ be the law of $\bW / \|\bW\|$ when $\bW \sim \Wish(n, r, 1)$.
    Then, the conclusion of Theorem~\ref{thm:mhc} applies to these $\mu_n$.
    That is, Algorithm~\ref{alg:mhc} runs with $\alpha = \gamma / \sqrt{rn}$ for a constant $\gamma > 0$ and a stream of $T$ i.i.d.\  $\bA_i \sim \mu_n$ outputs a stream of signs $x_1, \dots, x_T \in \{ \pm 1 \}$ that achieves, with high probability as $n \to \infty$,
    \begin{equation}
        \max_{t = 1}^T\left\|\sum_{s = 1}^t x_s \bA_s \right\| \lesssim \sqrt{rn} \log T.
    \end{equation}
\end{corollary}
\noindent
We note that this means of normalizing the $\bA_i$ to have norm 1 is natural, for example when $r = 1$ giving $\bA = \bu\bu^{\top}$ for $\bu$ uniformly random on the unit sphere, but the same result holds for any reasonable way of carrying out a normalization or truncation, by a straightforward application of our Proposition~\ref{prop:kaci-trunc}.
Because of this normalization this result deals with a slightly different distribution than Corollary~\ref{cor:lb-proj} above, but that result is straightforward to extend to allow for the same normalization as well.

We do not pursue analogous results for Wishart matrices with non-Gaussian underlying distributions here, but in principle these should be possible to obtain so long as one can prove corresponding anti-concentration inequalities for quadratic forms.
Results for some distributions, albeit a class not admitting a convenient description, would follow from Theorem~1.2 of \cite{Lovett-2010-GaussianPolynomialAntiConcentration}, for example.

\begin{rem}[Anti-concentration technicalities]
    Our reliance on the anti-concentration property introduces some nuances in the applicability of Theorem~\ref{thm:mhc}.
    For example, since the Theorem applies to matrices distributed as $\bu\bu^{\top}$ for $\bu \sim \Unif(\SS^{n - 1})$, one might also hope that the same would hold for $\bu \sim \tfrac{1}{\sqrt{n}}\Unif(\{\pm 1\}^n)$.
    However, this distribution does not satisfy a $\MACI(\eta)$ property for any $\eta > 0$, since its diagonal is constant.
    We expect that the result of Corollary~\ref{cor:wishart} should still hold for this distribution, but we will not attempt to address this technical difficulty here.
\end{rem}

\subsection{Related Work}
\label{sec:related}

\paragraph{Worst-case discrepancy}

The setting of Spencer's theorem (our Theorem~\ref{thm:spencer}) where $\bv_i \in [-1, 1]^n$ is only one of several worst-case settings---where we seek general bounds on the discrepancy of a sequence of vectors subject to some constraints---of interest in vector discrepancy theory.
Other well-studied restrictions include the cases of sparse $\bv_i \in \{0, 1\}^n$, corresponding to the \emph{Beck-Fiala conjecture}, and of $\|\bv_i\|_2 \leq 1$, corresponding to the \emph{\Komlos\ conjecture} (see the discussion in the general references \cite{Matousek-1999-GeometricDiscrepancy,Chazelle-2000-DiscrepancyMethod}).

One may consider the matrix versions of these constraints, which should be applied to the matrix eigenvalues.
Thus the ``matrix Beck-Fiala problem'' concerns the discrepancy of matrices with eigenvalues in $\{0, 1\}$ and with few non-zero eigenvalues, i.e., low-rank projection matrices as in our Corollary~\ref{cor:lb-proj}, while the ``matrix \Komlos\ problem'' concerns the discrepancy of matrices of bounded Frobenius norm.
For the former problem, to the best of our knowledge only the rank one case has been studied extensively, in large part due to its connection with the Kadison-Singer problem \cite{MSS-2014-RamanujanGraphsKadisonSinger,MSS-2015-InterlacingFamiliesKadisonSinger,KLS-2020-FourDeviationsRankOne}.
For the latter problem, \cite{DJR-2021-MatrixDiscrepancyMirrorDescent} showed that the natural matrix-valued generalization of the vector \Komlos\ conjecture does not hold; see also our discussion around Corollary~\ref{cor:ub-general}.

\paragraph{Average-case discrepancy}
The result of Theorem~\ref{thm:gaussian-vectors} holds for random vectors with i.i.d.\ entries subject to mild regularity conditions on the underlying distribution (Remark 2 of \cite{TMR-2020-BalancingGaussianVectors}).
Partial results on the Gaussian case preceding those of \cite{TMR-2020-BalancingGaussianVectors} include those of \cite{KKLO-1986-AverageCaseNumberPartitioning,Costello-2009-BalancingGaussianVectors}.

One important case to which these analyses does not apply is that of \emph{discrete} distributions, such as i.i.d.\ vectors $\bv_i \sim \Unif(\{\pm 1\}^n)$, which happens to also be the setting to which the strongest online algorithm proposed by \cite{BS-2020-OnlineBalancingRandom}---not the vector version of the MHC algorithm but an enhancement thereof---applies. The key issue here is that, for sufficiently large $T$, it will be possible to sign such vectors so that $\|\sum_{i = 1}^T x_i\bv_i\| \in \{0, 1\}$, so that the discrepancy is as small as possible subject to a parity constraint.
The threshold of $T$ for which this happens was established by \cite{ANW-2021-DiscrepancyRandomRectangular}; see the discussions in \cite{ANW-2021-DiscrepancyRandomRectangular,GKPX-2023-OGPDiscrepancy} for a review of the literature.

In the special scaling $T \sim n$, much more detailed information is known characterizing the typical discrepancy as a function of the ``aspect ratio'' $\alpha \colonequals \lim_{n \to \infty} T / n$.
This setting, up to changes of notation, goes by the name of the \emph{symmetric binary perceptron} model; see, e.g., \cite{APZ-2019-CapacitySymmetricBinaryPerceptron,PX-2021-IsingPerceptron1RSB,ALS-2022-ContiguityLimitSymmetricPerceptron,Altschuler-2022-CriticalWindowSymmetricPerceptron,SS-2023-SymmetricPerceptronThreshold}.

A few works have considered, rather than a fully average-case setting of i.i.d.\ random inputs, a \emph{smoothed analysis} of discrepancy problems where a small random perturbation is applied to the inputs.
As shown for different discrepancy problems by \cite{BJMSS-2021-PrefixDiscrepancySmoothed,BJMSS-2022-SmoothedAnalysisKomlos}, this perturbation can allow improvements on the best known discrepancy bounds.

To the best of our knowledge, there are no published results on average-case matrix discrepancy.
We note, however, that while this paper was being prepared we learned in private communications of parallel work by Afonso Bandeira and Antoine Maillard, performing calculations similar to those of Corollary~\ref{thm:goe} more precisely for the regime $T \sim n^2$ (the matrix-valued setting analogous to the symmetric binary perceptron model).

\paragraph{Online algorithms for discrepancy}
Online algorithms for discrepancy, including the vector version of the MHC algorithm, were first considered in a game-theoretic setting by \cite{Spencer-1977-BalancingGames} (here, the first player sends vectors to the second, who must choose signs for them once they are received; thus, the first player generates an adversarial instance of a discrepancy problem, while the second executes an online algorithm).
The results of \cite{BS-2020-OnlineBalancingRandom} as cited in Theorem~\ref{thm:vector-online} analyze this algorithm and a variation thereof in the average-case setting.
More specifically, the algorithm achieving Part 1 of Theorem~\ref{thm:vector-online} combines the vector version of the MHC algorithm with an adjustment involving taking a majority vote over vector coordinates.
We have not found a straightforward matrix version of this that would allow us to improve on the MHC algorithm.
The analysis of \cite{BJSS-2020-OnlineVectorBalancingDiscrepancy}, which yields Part 2 of Theorem~\ref{thm:vector-online}, allows much more general vector distributions but incurs the cost of $\log(T)$ in the discrepancy bound, as in our Theorem~\ref{thm:mhc}.
We leave to future work the interesting open question of whether these factors may be removed in the matrix case by either a sharper analysis of the MHC algorithm or a ``hybrid'' algorithm like that of \cite{BS-2020-OnlineBalancingRandom}.

The work of \cite{ALS-2021-DiscrepancySelfBalancingWalk} we have mentioned above proposed an alternative algorithmic idea for the \emph{oblivious} adversarial setting, where the adversarial online vector inputs do not depend on the state (and, in particular, the previous sign choices) of the algorithm.
They defined a way of choosing random signs $x_t$ that yields a discrepancy bound that is the same as \cite{BS-2020-OnlineBalancingRandom,BJSS-2020-OnlineVectorBalancingDiscrepancy} except for a logarithmic factor of $\log(nT)$.
As we discuss in Remark~\ref{rem:als}, because their argument shows subgaussianity of the partial sums generated by the algorithm, their approach can be used for matrix discrepancy, but the MHC algorithm gives a slightly better bound.
The related algorithm of \cite{LSS-2021-GaussianFixedPointWalk}, treats a different but similar problem in a similar way.
Finally, \cite{KRR-2023-OnlineDiscrepancyMinimization} shows that an optimal bound of $O(\sqrt{\log T})$ can be achieved against an oblivious adversary (again via subgaussianity of the partial sums), but only by an inefficient online algorithm.
It remains an open problem whether the same can be achieved in polynomial time.

The only prior work we are aware of that explicitly attempts to extend these ideas to the matrix case is that of \cite{Zouzias-2011-MatrixHyperbolicCosine}, which describes Algorithm~\ref{alg:mhc} and provides a general analysis.
However, the application of their main result (Theorem 4 of \cite{Zouzias-2011-MatrixHyperbolicCosine}) in our setting is weaker; in the setting of Theorem~\ref{thm:mhc}, that analysis would only show that Algorithm~\ref{alg:mhc} produces a signing with
\begin{equation}
    \left\|\sum_{t = 1}^T x_t \bA_t \right\| \lesssim \sqrt{T \log n},
\end{equation}
the same as achieved by random signs per \eqref{eq:random-signs}.
In particular, once $T \gg n$ this analysis is severely suboptimal.
We also remark that a similar potential function on matrices was used, albeit only as a proof technique, by \cite{SV-2013-CovarianceEstimation2Moments}.

\paragraph{Handling non-commutativity}
Our approach to the analysis of the MHC algorithm is similar in spirit to a long line of work on matrix concentration inequalities. The analysis of the corresponding vector algorithm by \cite{BS-2020-OnlineBalancingRandom} may be viewed as treating the case of $\bA_1, \dots, \bA_T$ commutative matrices, and, ultimately, we must confront the issue of non-commutativity in the form of the failure of the identity $\exp(\bX + \bY) \neq \exp(\bX)\exp(\bY)$ for matrices.
As in the foundational work of \cite{AW-2002-IdentificationQuantumChannels} (see also the survey \cite{Tropp-2015-MatrixConc}) on matrix concentration, this may be circumvented by using the \emph{Golden-Thompson inequality} to relate these two quantities.

\section{Discrepancy Asymptotics}

We first introduce our proof techniques for upper and lower bounds on matrix discrepancy.

\subsection{Lower Bounds: First Moment Method and Scaling Heuristic}
\label{sec:lb}

In the first moment method, we will consider the number of signings that achieve a particular discrepancy $\delta$.
However, it will also be useful to transform the matrices $\bA_i$ involved into random matrices of mean zero (what is called ``symmetrization'' in similar situations in probability theory, not to be confused with symmetrization of asymmetric matrices), for which we introduce i.i.d.\ $y_1, \dots, y_T \sim \Unif(\{ \pm 1\})$.
We then have
\begin{align*}
    N_{\delta}
    &\colonequals \#\left\{ \bx \in \{\pm 1\}^T: \left\|\sum_{i = 1}^T x_i \bA_i \right\| \leq \delta \right\} \\
    &= \#\left\{ \bx \in \{\pm 1\}^T: \left\|\sum_{i = 1}^T x_i y_i\bA_i \right\| \leq \delta \right\} \\
    &= \sum_{\bx \in \{ \pm 1\}^T} \One\left\{\left\|\sum_{i = 1}^T x_i y_i\bA_i \right\| \leq \delta \right\}. \numberthis
\end{align*}
Taking expectations, we then find
\begin{equation}
    \EE N_{\delta} = 2^T \PP\left[\left\|\sum_{i = 1}^T y_i\bA_i \right\| \leq \delta\right]. \label{eq:1mm}
\end{equation}
If $\EE N_{\delta} \to 0$ for some $\delta = \delta(n)$ as $n \to \infty$, then by Markov's inequality we will have that $\Delta(\bA_1, \dots, \bA_T) \geq \delta$ with high probability.

Let us give a heuristic calculation to establish what scaling we expect for the smallest $\delta$ for which the above argument works, substantiating Conjecture~\ref{conj:lb}.
Recall that the $y_i\bA_i$ are i.i.d.\ centered random matrices of norm at most 1 and rank $r$.
We expect that sums of roughly $n / r$ of the $y_i\bA_i$ should be full-rank matrices, still of norm at most 1.
Thus, view our sum as
\begin{equation}
    \sum_{i = 1}^T y_i\bA_i = \sum_{j = 1}^{rT / n} \bB_j,
\end{equation}
where the $\bB_j$ are again i.i.d.\ centered random matrices of norm $O(1)$ and full rank, each distributed as the sum of $n / r$ independent copies of $y_1\bA_1$. Now, by the free central limit theorem (see, e.g., \cite{VDN-1992-FreeRandomVariables}), we expect the eigenvalues of such a matrix to be close to those of $\sqrt{rT / n} \, \bW$ for $\bW \sim \GOE(n)$.
Thus, the probability we are interested in should be approximated by a large deviation probability for the GOE,
\begin{align*}
    \PP\left[\left\|\sum_{i = 1}^T s_i\bA_i \right\|
    \leq \delta\right]
    &\approx \Px_{\bW \sim \mathrm{GOE}}\left[\|\bW\| \leq \sqrt{\frac{n}{rT}}\delta \right]
    \intertext{and for this kind of probability we will develop an estimate below (Corollary~\ref{cor:goe-small-ball}) based on the explicit formula for the density of the GOE eigenvalues, which will give that the leading order behavior is}
    &\approx \left(C\sqrt{\frac{n}{rT}}\delta\right)^{n^2/2} \numberthis
\end{align*}
for some $C > 0$. We find that the first moment is
\begin{equation}
    \EE N_{\delta} \approx 2^T\left(C\sqrt{\frac{n}{rT}}\delta\right)^{n^2/2},
\end{equation}
and this would give that $N_{\delta} = 0$ with high probability for any $\delta \leq c\sqrt{\frac{rT}{n}} 4^{-T/n^2}$ for $c < C^{-1}$.
So, we expect a high probability lower bound of the form
\begin{equation}
    \Delta(\bA_1, \dots, \bA_T) \gtrsim \sqrt{\frac{rT}{n}} 4^{-T/n^2},
\end{equation}
as we claimed in Conjecture~\ref{conj:lb}.

\subsection{Lower Bounds: Gramian Spectral Method}

We will also find useful a different approach to lower bounds, which amounts to treating the matrices involved as vectors and comparing with an $\ell^2$ vector discrepancy problem.
This gives a spectral lower bound on $\Delta(\bA_1, \dots, \bA_T)$ involving the Gram matrix of the $\bA_i$ under the Frobenius or trace inner product.

\begin{lemma}[Spectral discrepancy lower bound]
    \label{lem:spectral-lb}
    Let $\bA_1, \dots, \bA_T \in \RR^{n \times n}_{\sym}$.
    Let $\bM \in \RR^{T \times T}_{\sym}$ have $M_{ij} \colonequals \langle \bA_i, \bA_j \rangle$.
    Then,
    \begin{equation}
        \Delta(\bA_1, \dots, \bA_T) \geq \sqrt{\frac{T}{n} \lambda_{\min}(\bM)}.
    \end{equation}
\end{lemma}
\begin{proof}
    Suppose $\bx \in \{\pm 1\}^T$.
    We then have
    \begin{align*}
        \left\|\sum_{i = 1}^T x_i \bA_i \right\|^2
        &\geq \frac{1}{n}\left\|\sum_{i = 1}^T x_i \bA_i \right\|_F^2 \\
        &= \frac{1}{n}\sum_{i, j = 1}^T x_i x_j \langle \bA_i, \bA_j \rangle \\
        &= \frac{1}{n}\bx^{\top} \bM \bx \\
        &\geq \frac{1}{n} \|\bx\|^2 \lambda_{\min}(\bM) \\
        &= \frac{T}{n}\lambda_{\min}(\bM),
    \end{align*}
    and taking square roots on either side completes the proof.
\end{proof}

\subsection{Upper Bounds: Vectorization}

Our approach to asymptotic discrepancy upper bounds will ultimately take a similar form to the above spectral lower bound, but will be achieved by a quite different analysis.

\begin{definition}[Subgaussian vector]
    A random vector $\bv \in \RR^n$ is \emph{$\sigma^2$-subgaussian} if, for all $\by \in \RR^n$,  $\EE \exp(\langle \bv, \by \rangle) \leq \exp(\frac{1}{2}\sigma^2\|\by\|^2)$.
\end{definition}
\noindent
One way to view the definition is that $\bv$ is ``subgaussian in all directions,'' in the sense that $\langle \bv, \by \rangle$ is $\sigma^2$-subgaussian for all unit vectors $\by$.

\begin{proposition}[Theorem 1 of \cite{BDGL-2018-GramSchmidtWalk}]
    \label{prop:gsw}
    Let $\bv_1, \dots, \bv_T \in \RR^n$ have $\|\bv_i\| \leq 1$.
    Then, there are random $x_1, \dots, x_T \in \{\pm 1\}^n$ such that, for each $t \in [T]$, $\sum_{s = 1}^t x_s \bv_s$ is $\sigma^2$-subgaussian for an absolute constant $\sigma^2$ (one may take $\sigma^2 = 40$).
    Moreover, these $x_i$ may be computed in polynomial time by an offline algorithm.
\end{proposition}
\noindent
Some related work is worth mentioning.
First, the analysis of \cite{HSSZ-2019-BalancingCovariatesGramSchmidt} sharpened the subgaussianity constant to $\sigma^2 = 1$ (which, as they also show, is the best possible).
Next, \cite{KRR-2023-OnlineDiscrepancyMinimization} show that this subgaussianity can be achieved by an online algorithm, but this algorithm does not run in polynomial time.
A polynomial-time (and surprisingly simple) algorithm has been proposed by \cite{ALS-2021-DiscrepancySelfBalancingWalk}, which only gives $O(\log(nT  / \delta)^2)$-subgaussianity, on an event of probability $1 - \delta$ for some $\delta > 0$ (see also the related algorithm of \cite{LSS-2021-GaussianFixedPointWalk}, which treats a different but discrepancy-like problem).
Having the best of all worlds---a polynomial-time online algorithm producing an $O(1)$-subgaussian output---remains an open problem.

Perhaps surprisingly, this vector-valued result is already enough to derive powerful bounds on matrix discrepancy, because a standard technique allows us to control the operator norm of a random matrix whose vectorization is subgaussian.

\begin{proposition}
    \label{prop:subgauss-mx}
    Suppose that $\bM \in \RR^{n \times n}_{\sym}$ is a random matrix such that $\symvec(\bM)$ is $\sigma^2$-subgaussian.
    Then,
    \begin{equation}
        \PP[\|\bM\| \geq 4\sigma \sqrt{n} ] \leq \exp(-n).
    \end{equation}
\end{proposition}
\begin{proof}
    By the assumption, for any $\theta > 0$ and $\by \in \RR^n$ with $\|\by\| = 1$,
\begin{align*}
  \mathbb{P}[\by^{\top}\bM \by \geq t]
  &\leq \frac{\mathbb{E} \exp(\langle \theta \by\by^{\top}, \bM \rangle)}{\exp(\theta t)} \\
  &\leq \exp\left(\frac{\sigma^2}{2} \|\theta \by\by^{\top}\|_F^2 - \theta t\right) \\
  &= \exp\left(\frac{\sigma^2}{2} \theta^2 - \theta t\right).
\end{align*}
Taking $\theta = t / \sigma^2$, we have
\begin{equation}
    \mathbb{P}[\by^{\top}\bM \by \geq t] \leq \exp\left(-\frac{t^2}{2\sigma^2}\right).
\end{equation}
    Now, let $\sS \subset \RR^n$ be an $\epsilon$-net of the unit sphere for some $\epsilon \in (0, 1)$.
    A standard argument (Lemma 2.3.2 of \cite{Tao-2012-RandomMatrixTheory}) shows that
    \begin{equation}
        \|\bM\| \leq \frac{1}{1 - \epsilon} \max_{\by \in \sS} \by^{\top} \bM \by,
    \end{equation}
    Thus we have, taking $\epsilon = \frac{1}{2}$,
    \begin{align*}
        \PP[\|\bM\| \geq C\sqrt{n}]
        &\leq \PP\left[\max_{\bv \in \sS} \by^{\top}\bM\by \geq \frac{C}{2}\sqrt{n}\right] \\
        &\leq \exp\left(-\frac{C^2}{8\sigma^2}n\right) |\sS|
        \intertext{And by another standard fact (Lemma 2.3.4 of \cite{Tao-2012-RandomMatrixTheory}) it is possible to take $|\sS| \leq 3^n$, whereby}
        &\leq \exp\left(\left(\log(3) - \frac{C^2}{8\sigma^2}\right)n\right),
    \end{align*}
    and taking $C = 4\sigma$ gives the result.
\end{proof}

\begin{corollary}
    \label{cor:ub-general}
    For any $\bA_1, \dots, \bA_T \in \RR^{n \times n}_{\sym}$, $\Delta(\bA_1, \dots, \bA_T) = O(\sqrt{n} \cdot \max_i \|\bA_i\|_F)$.
\end{corollary}
\begin{proof}
    Let $F \colonequals \max_i \|\bA_i\|_F$.
    Proposition~\ref{prop:gsw} on $\symvec(\bA_1) / F, \dots, \symvec(\bA_T) / F$ gives that there are random $x_1, \dots, x_T \in \{ \pm 1\}^T$ such that $\sum_{i = 1}^T x_i\bA_i$ is $O(F^2)$-subgaussian.
    The result then follows from Proposition~\ref{prop:subgauss-mx}.
\end{proof}
\noindent
As Corollary 1.8 of \cite{DJR-2021-MatrixDiscrepancyMirrorDescent} shows, for $\|\bA_i\|_F$ equal for all $i$, this is the best possible general operator norm discrepancy bound (and thus the best possible version of a ``matrix \Komlos\ conjecture'' result).
In fact, our results below will give an alternative proof of their lower bound, showing that even i.i.d.\ random GOE matrices achieve it (when $T \sim n^2$).

When $T \sim n^2$, we note that the upper bound of Corollary~\ref{cor:ub-general} is a natural counterpart to the lower bound of Lemma~\ref{lem:spectral-lb}.
Indeed, by loosening the upper bound, we may bound on both sides by the spectrum of $\bM$ the Gram matrix of the $\bA_i$ as
\begin{equation}
    \sqrt{\frac{T}{n}\lambda_{\min}(\bM)} \leq \Delta(\bA_1, \dots, \bA_T) \leq \sqrt{n \lambda_{\max}(\bM)}.
\end{equation}
When $T \sim n^2$, the initial factors are both $\Theta(n)$.
Thus, in this regime, so long as $\bM$ is well-conditioned---having smallest and largest eigenvalues of the same order---these two rather naive bounds alone will give a description of the discrepancy that is tight up to constants.
This is the general idea that will drive our proof of Theorem~\ref{thm:lb}.

\begin{rem}[Second moment method]
    \label{rem:2mm}
    The approach we take to upper bounds is rather unsatisfying: it does not give accurate results except in the regime $T \sim n^2$, and even there is likely not tight at the level of multiplicative constants.
    We suspect that better upper bounds could be obtained via the second moment method, a natural counterpart to the first moment method we use for some of our lower bounds.
    However, this seems quite technically challenging even for $\bA_i \sim \GOE(n)$: while for the first moment method one must understand the probability of a GOE matrix having unusually small operator norm, which may be treated with large deviations theory when $T \sim n^2$ to obtain more precise results than ours, the second moment method leads to considering the probability that \emph{two} correlated GOE matrices \emph{both} have small operator norm. This seems to be a much more challenging problem; we propose the pursuit of this as a natural direction for future work.
\end{rem}

\subsection{GOE Matrices: Proof of Theorem~\ref{thm:goe}}

We isolate the main technical part of our argument, an estimate on small operator norm ball probabilities for the GOE, which may be of independent interest.
First, we give the following lower bound on the Vandermonde determinant with parameters confined to an interval.
We learned of this bound from the online discussion \cite{Pinelis-2017-VandermondePost}, but include a proof for the sake of completeness.

\begin{lemma}
    For any $n \geq 1$,
    \label{lem:vandermonde}
    \begin{equation}
        \sup_{0 \leq \lambda_1 \leq \cdots \leq \lambda_n \leq 1} \prod_{1 \leq i < j \leq n}|\lambda_i - \lambda_j| = \prod_{j = 0}^{n - 1} \frac{j^j(j + 1)^{(j + 1)/2}}{(j + n - 1)^{(j + n - 1)/2}} = 2^{-(1 + o(1))n^2}.
    \end{equation}
\end{lemma}
\begin{proof}
    Let us denote the Vandermonde determinant by:
    \begin{equation}
        V(\bm\lambda) \colonequals \prod_{1 \leq i < j \leq n}|\lambda_i - \lambda_j|.
    \end{equation}
    Then, the classical Selberg integral (see, e.g., Equation (1.1) of \cite{Forrester-2008-SelbergIntegral}) gives
    \begin{align*}
        \int_{[0, 1]^n} V(\bm\lambda)^{2\beta} d\bm\lambda
        &= \prod_{j = 0}^{n - 1}\frac{\Gamma(1 + j\beta)^2\Gamma(1 + (j + 1)\beta)}{\Gamma(2 + (n + j - 1)\beta)\Gamma(1 + \beta)} \\
        &= n! \prod_{j = 0}^{n - 1}\frac{\Gamma(1 + j\beta)^2\Gamma((j + 1)\beta)}{\Gamma(2 + (n + j - 1)\beta)\Gamma(\beta)}.
    \end{align*}
    We then have, using that $\lim_{\beta \to \infty} \beta^{-y}\Gamma(x + y\beta)^{1/\beta} = (y / e)^y$ by Stirling's approximation,
    \begin{align*}
        &\hspace{-0.25cm}\sup_{0 \leq \lambda_1 \leq \cdots \leq \lambda_n \leq 1} \prod_{1 \leq i < j \leq n}|\lambda_i - \lambda_j| \\
        &= \lim_{\beta \to \infty} \left(\int_{[0, 1]^n} V(\bm\lambda)^{2\beta} d\bm\lambda \right)^{1 / 2\beta} \\
        &= \prod_{j = 0}^{n - 1}\frac{j^j(j + 1)^{(j + 1)/2}}{(j + n - 1)^{(j + n - 1)/2}}
        \intertext{giving the first formula. For the second estimate, we rewrite}
        &= \exp\left(\sum_{j = 0}^{n - 1} \left(j\log(j) + \frac{j + 1}{2}\log(j + 1) - \frac{j + n - 1}{2} \log(j + n - 1)\right)\right) \\
        &= \exp\left(n\left[\sum_{j = 0}^{n - 1} \left(\frac{j}{n}\log\left(\frac{j}{n}\right) + \frac{j + 1}{2n}\log\left(\frac{j + 1}{n}\right) - \left(\frac{1}{2} + \frac{j - 1}{2n}\right) \log\left(1 + \frac{j - 1}{n}\right)\right)\right]\right)
        \intertext{and using standard estimates on Riemann sums,}
        &= \exp\left((1 + o(1))n^2\left[\frac{3}{2} \int_0^1 x\log x \, dx - \int_0^1 \frac{1 + x}{2} \log(1 + x) \, dx\right]\right) \\
        &= \exp\left((1 + o(1))n^2\left[-\frac{3}{2} \cdot \frac{1}{4} + - \log(2) + \frac{3}{8}\right]\right) \\
        &= 2^{-(1 + o(1))n^2},
    \end{align*}
   as claimed.
\end{proof}

\begin{corollary}
    \label{cor:goe-small-ball}
    For any $\delta = \delta(n) > 0$,
    \begin{equation}
        \Px_{\bW \sim \GOE(n, 1)}[\|\bW\| \leq \delta] \leq \left(\frac{e^{3/4}}{2\sqrt{n}} \delta\right)^{(1 + o(1))n^2/2}.
    \end{equation}
\end{corollary}

\begin{proof}
    For $\bW \sim \GOE(n, 1)$ with ordered eigenvalues $\lambda_1 \geq \cdots \geq \lambda_n \in \RR$, the joint density of $(\lambda_1, \dots, \lambda_n)$ is known \cite[Theorem 2.5.2]{AGZ-2010-RandomMatrices} to be
    \begin{equation}
        \rho(\lambda_1, \dots, \lambda_n) = \One\{\lambda_1 \geq \cdots \geq \lambda_n\} \cdot C_n \exp\left(-\frac{1}{4}\sum_{i = 1}^n \lambda_i^2\right)\prod_{1 \leq i < j \leq n} |\lambda_i - \lambda_j|,
    \end{equation}
    where
    \begin{align*}
        C_n
        &= (2\pi)^{-n/2} 2^{-n(n + 1)/4} \prod_{j = 1}^n \frac{\Gamma(1/2)}{\Gamma(j/2)} \\
        &= 2^{-n(n + 3)/4} \prod_{j = 1}^n \frac{1}{\Gamma(j/2)}
        \intertext{where we have used that $\Gamma(1/2) = \sqrt{\pi}$. Now, using that, for $x \geq 1$, $\Gamma(x) \geq (\frac{x -1}{e})^{x - 1}$, we find}
        &\leq \sqrt{2\pi} \, 2^{-n(n + 3)/4} \exp\left(-\sum_{j = 0}^{n - 2} \frac{j}{2} \log\left(\frac{j}{2e}\right)\right) \\
        &\leq \sqrt{2\pi} \, 2^{-n(n + 3)/4 + (n - 2)(n - 1)/4} \exp\left(-\frac{1}{2}\sum_{j = 1}^{n - 2} j \log j + \frac{(n - 2)(n - 1)}{4}\right) \\
        &\leq  \sqrt{2\pi} \, 2^{-(3n - 1)/2} \exp\left(-\frac{1}{2}\int_{1}^{n - 2} x \log x \, dx + \frac{(n - 2)(n - 1)}{4}\right) \\
        &= \sqrt{2\pi} \, 2^{-(3n - 1)/2} \exp\left(-\frac{1}{4}(n - 2)^2\log(n - 2) + \frac{1}{8}(n - 2)^2 + \frac{(n - 2)(n - 1)}{4} + \frac{1}{4}\right)
        \intertext{and, estimating roughly,}
        &\leq (e^{3/4}\sqrt{n})^{-(1 + o(1))n^2/2}. \numberthis
    \end{align*}

    Using this and Lemma~\ref{lem:vandermonde}, we find that, whenever $\delta \geq \lambda_1 \geq \cdots \geq \lambda_n \geq -\delta$, then
    \begin{equation}
        \rho(\lambda_1, \dots, \lambda_n) \leq (e^{3/4}\sqrt{n})^{-(1 + o(1))n^2/2} \cdot (2\delta)^{n(n - 1)/2} \cdot 2^{-(1 + o(1))n^2} = \left(\frac{e^{3/4}}{2\sqrt{n}}\delta\right)^{(1 + o(1))n^2/2}.
    \end{equation}
    Since the volume of the set of such $(\lambda_1, \dots, \lambda_n)$ is $\frac{1}{n!}(2\delta)^n$ which is at most 1 for any $\delta$ for which the claim is not vacuously true, the result follows.
\end{proof}

\begin{rem}[Tighter estimates from large deviations principles]
    A more precise analysis of this probability in the special scaling $\delta = \what{\delta} \sqrt{n}$ may be obtained by the large deviations principle for the empirical spectral distribution of the GOE of \cite{BAG-1997-WignerLargeDeviations}.
    This would give the precise function $f(\what{\delta})$ such that, to leading order, $\PP[\|\bW\| \leq \what{\delta}\sqrt{n} ] \approx f(\what{\delta})^{n^2/2}$.
    When $T = cn^2$, this would give a bound on the constant $g(c)$ such that $\Delta(\bA_1, \dots, \bA_T) \sim g(c)n$.
    We do not pursue this more careful analysis here since our upper bounds are loose anyway, but propose the determination of $g(c)$ as an interesting open problem.
    See also our discussion of the second moment method in Remark~\ref{rem:2mm}.
    By comparison, our approach has the advantage of also applying to the case of ``very large deviations'' when $\delta \ll \sqrt{n}$.
\end{rem}

\begin{rem}[Alternative estimates]
    Another way to bound this ``small ball probability'' of the GOE is to make the following observations, which David Gamarnik brought to our attention after this paper was posted online.
    First, the density of the GOE is maximized at the zero matrix, where it is at most $c^{n(n + 1)/2}$ for a $c > 0$.
    Thus $\PP[\|\bW\| \leq \delta] \leq (c\delta)^{n(n + 1)/2} V_n$, where $V$ is the volume of the unit ball under the operator norm in $\RR^{n \times n}_{\sym}$.
    The unit ball under the operator norm is contained in the ball of radius $\sqrt{n}$ under the Frobenius norm, which is just the ordinary $\ell^2$ norm in this space (up to a constant factor).
    Thus $V_n \leq \sqrt{n}^{n(n + 1)/2} n^{-(1 + o(1))n(n + 1)/2} \leq n^{-(1 + o(1))n(n + 1)/4}$, and a similar result to ours up to constants follows from combining these.
    We give the above approach instead since, as mentioned above, it is more likely to be possible to improve to give bounds that are sharp even at the level of constant factors.
\end{rem}

\begin{proof}[Proof of Theorem~\ref{thm:goe}]
For the upper bound, we just note that, for $\bW \sim \GOE(n)$, $\|\bW\|_F = O(\sqrt{n})$ with high probability by standard concentration inequalities.
The result then follows by Corollary~\ref{cor:ub-general}.

For the lower bound, starting from \eqref{eq:1mm}, we note that $\sum_{i = 1}^T y_i \bA_i$ has \emph{exactly} the same law as $\sqrt{T / n} \bW$ for $\bW \sim \GOE(n)$.
Thus we have, for a constant $C > 0$,
\begin{align*}
    \EE N_{\delta}
    &= 2^T \PP\left[\|\bW\| \leq \delta\sqrt{\frac{n}{T}}\right] \\
    &\leq 2^T \left(\frac{C\delta}{\sqrt{T}}\right)^{n^2/2}
    \intertext{and, when $\delta \leq (C + \epsilon)^{-1}\sqrt{T} 4^{-T/n^2}$ for some fixed $\epsilon > 0$, then}
    &\leq \left(\frac{C}{C + \epsilon}\right)^{n^2/2}, \numberthis
\end{align*}
which tends to zero as $n \to \infty$.
\end{proof}

\subsection{General Analysis: Proof of Theorem~\ref{thm:lb}}

We use the spectral lower bound on the discrepancy of Lemma~\ref{lem:spectral-lb}, together with the following powerful general bound on the spectrum of Gram matrices of independent vectors.

\begin{definition}[Orlicz norm]
    Let $\ba \in \RR^N$ be a random vector.
    We define its \emph{$\psi_{1}$ norm} to be:
    \begin{equation}
        \|\ba\|_{\psi_{1}} \colonequals \sup_{\by \in \SS^{N - 1}} \inf\left\{C > 0: \EE\left[\exp\left(\frac{|\langle \ba, \by \rangle|}{C}\right)\right] \leq 2\right\}.
    \end{equation}
\end{definition}

\begin{proposition}[Theorem 3.3 of \cite{ALPTJ-2011-RIPIndependentColumns}]
    \label{prop:alptj}
    Let $T \leq N$, $\ba_1, \dots, \ba_T$ be independent centered random vectors, each with finite $\psi_{1}$ norm, and let $\psi \colonequals \max_{i \in [T]} \|\ba_i\|_{\psi_{1}}$.
    Let $\bM \in \RR^{T \times T}_{\sym}$ have entries $M_{ij} = \langle \ba_i, \ba_j \rangle$.
    There exist universal constants $C_1, C_2 > 0$ such that, for any $\theta \in (0, 1)$, we have
    \begin{equation}
        1 - \delta \leq \lambda_{\min}(\bM) \leq \lambda_{\max}(\bM) \leq 1 + \delta,
    \end{equation}
    where
    \begin{equation}
        \delta = C_1(\psi \sqrt{N} + \sqrt{1 + \theta})^2 \sqrt{\frac{T}{N}}\left(1 + \frac{1}{2}\log\left(\frac{N}{T}\right)\right) + \theta,
    \end{equation}
    with probability at least
    \begin{equation}
        1 - C_1 \exp\left(-C_2 \sqrt{T}\left(1 + \frac{1}{2}\log\left(\frac{N}{T}\right)\right)\right) - 2\PP\left[ \max_{i \in [T]} \big| \|\ba_i\|^2 - 1\big| \geq \theta\right].
    \end{equation}
\end{proposition}

\begin{proof}[Proof of Theorem~\ref{thm:lb}]
    Note that, under Condition~1 of the Theorem, we have $\|\bA_i\|_F^2 \in [\frac{1}{2}r, \frac{3}{2}r]$ for all $i \in [n]$ with high probability, by a union bound.
    Thus the upper bound in the Theorem is immediate from Corollary~\ref{cor:ub-general}.

    For the lower bound, we will use Lemma~\ref{lem:spectral-lb} together with the above tools.
    Let us use Proposition~\ref{prop:alptj} with $\ba_i \colonequals \symvec(\bA_i) / \sqrt{r}$, $N = \frac{n(n + 1)}{2}$ (the dimension of the $\ba_i$), and $\theta = \frac{1}{2}$.
    By the above observation, since $\|\ba_i\|^2 = \|\bA_i\|_F^2 / r \in [\frac{1}{2}, \frac{3}{2}]$ with high probability, we have $\PP\left[ \max_{i \in [T]} \big| \|\ba_i\|^2 - 1\big| \geq \theta\right] = o(1)$.
    By Condition 2 of the Theorem, we have $\psi = O(1/n) = O(1 / \sqrt{N})$.
    Taking $T = \epsilon n^2$ for sufficiently small $\epsilon$, we can then ensure that, say, $\lambda_{\min}(\bM) \geq \frac{1}{3}$ with high probability.

    Since $\langle \ba_i, \ba_j \rangle = \frac{1}{r} \langle \bA_i, \bA_j \rangle$, using this with Lemma~\ref{lem:spectral-lb} then shows that $\Delta(\bA_1, \dots, \bA_T) \geq \sqrt{3rT / n} = \sqrt{3\epsilon rn}$ with high probability, completing the proof.
\end{proof}

\subsection{Wishart Matrices: Proof of Corollary~\ref{cor:lb-proj}}

As mentioned earlier, to prove this result we will need to make a slight modification to the argument above for the general Theorem~\ref{thm:lb}.

\begin{proof}[Proof of Corollary~\ref{cor:lb-proj}]
For the sake of clarity, let us consider $\bG_i$ having i.i.d.\ entries distributed as $\sN(0, 1)$ instead of $\sN(0, \frac{1}{n})$.
Then, the statement concerns $\bA_i = \frac{1}{n}\bG_i\bG_i^{\top}$.

Note that $\EE \bG_i\bG_i^{\top} = r \bm I_n$.
Instead of working directly with $\bG_i\bG_i^{\top}$, we will work first with the centered version, $\bA_i^{(0)} = \frac{1}{n}(\bG_i\bG_i^{\top} - r \bm I_n)$, and recover the stated version from a result on this one.

We will use an intermediate result from the proof of Theorem~\ref{thm:lb} for these $\bA_i^{(0)}$.
Let us first verify that the Theorem applies.

For Condition 1, though it is possible to analyze more directly, we may also reuse Proposition~\ref{prop:alptj}.
This implies, when applied to the $\bG_i$, that, with high probability, for any $\delta > 0$, $1 - \delta \leq \lambda_{\min}(\frac{1}{n} \bG_i^{\top} \bG_i) \leq \lambda_{\max}(\frac{1}{n}\bG_i^{\top}\bG_i) \leq 1 + \delta$ for all $i \in [n]$ (using here our assumption that $r \ll n$).
The required bound on the Frobenius norms then follows immediately.

For Condition 2, let us work with a single $\bG \in \RR^{n \times r}$ with i.i.d.\ standard Gaussian entries.
It suffices to show that, for some absolute constant $c > 0$ and any $\bY \in \RR^{n \times n}_{\sym}$ with $\|\bY\|_F = 1$, we have
\begin{equation}
    \EE \exp\left(c\frac{n}{\sqrt{r}} \left|\left\langle \frac{1}{n}(\bG\bG^{\top} - r \bm I_n), \bY \right\rangle\right|\right) \stackrel{?}{\leq} 2.
\end{equation}
Let us write $\bg_1, \dots, \bg_r \in \RR^n$ for the (random standard Gaussian vector) columns of $\bG$.
First, observe that we have the general inequality $\exp(|t|) \leq \frac{2}{3}(\exp(2t) + \exp(-2t))$ for all $t \in \RR$.
Thus let us control first the same expression without absolute values, and with $c$ multiplied by 2.
We have:

\begin{align*}
    &\hspace{-1cm}\EE \exp\left(2c\frac{n}{\sqrt{r}} \left\langle \frac{1}{n}(\bG\bG^{\top} - r \bm  I_n), \bY \right\rangle\right) \\
    &= \EE \exp\left(\frac{2c}{\sqrt{r}} \langle \bG\bG^{\top} - r \bm I_n, \bY \rangle\right) \\
    &= \EE \exp\left(\frac{2c}{\sqrt{r}} \left\langle \sum_{i = 1}^r \left(\bg_i\bg_i^{\top} - \bm I_n\right), \bY \right\rangle\right) \\
    &= \left(\exp\left(-\frac{2c}{\sqrt{r}} \Tr(\bY)\right) \EE \exp\left(\frac{2c}{\sqrt{r}} \bg_1^{\top} \bY \bg_1 \right)\right)^r \\
    &= \left(\exp\left(-\frac{2c}{\sqrt{r}} \Tr(\bY)\right) \det\left(\bm I_n - \frac{4c}{\sqrt{r}} \bm Y\right)^{-1/2}\right)^r \\
    &= \exp\left(-2c\sqrt{r}\Tr(\bY) - \frac{r}{2}\sum_{i = 1}^n \log\left(1 - \frac{4c}{\sqrt{r}}\lambda_i(\bY)\right)\right)
    \intertext{Let us take $c \leq \frac{1}{8}$, so that $\frac{4c}{\sqrt{r}}\lambda_i(\bY) \leq \frac{1}{2}$ for any $r \geq 1$ (recalling that $\|\bY\| \leq \|\bY\|_F = 1$ by assumption).
    One may check that $-\log(1 - t) \leq t + t^2$ for all $0 < t \leq \frac{1}{2}$, whereby}
    &\leq \exp\left(-2c\sqrt{r}\Tr(\bY) + \frac{r}{2}\sum_{i = 1}^n \left(\frac{4c}{\sqrt{r}}\lambda_i(\bY) + \frac{16c^2}{r}\lambda_i(\bY)^2\right)\right)
    \intertext{and the first two terms cancel, leaving}
    &= \exp\left(8c^2\|\bY\|_F^2\right) \\
    &= \exp(8c^2).
\end{align*}
Now, negating the quantity in the exponential amounts to replacing $\bY$ by $-\bY$, which has no effect on the above calculation.
Thus we find
\begin{equation}
\EE \exp\left(c\frac{n}{\sqrt{r}} \left|\left\langle \frac{1}{n}(\bG\bG^{\top} - r \bm I_n), \bY \right\rangle\right|\right) \leq \frac{4}{3}\exp(8c^2),
\end{equation}
and taking $c > 0$ small enough will then show that Condition 2 of the Theorem is satisfied.

Instead of using the Theorem directly, however, we will use that these conditions allow us to apply Lemma~\ref{lem:spectral-lb}, as used in the proof of the Theorem.
Let $\bM^{(0)} \in \RR^{T \times T}_{\sym}$ have $M_{ij}^{(0)} = \langle \bA_i^{(0)}, \bA_j^{(0)}\rangle$.
Via Proposition~\ref{prop:alptj}, we then have that, for $T = \epsilon n^2$ for $\epsilon > 0$ sufficiently small, $\lambda_{\min}(\bM^{(0)}) \geq \frac{1}{2}r$.

Next, let $\bA_i \colonequals \frac{1}{n}\bG_i\bG_i^{\top}$ be the matrices we are actually interested in, and let $M_{ij} \colonequals \langle \bA_i, \bA_j \rangle$.
We will control $\lambda_{\min}(\bM)$ by reducing it to our analysis of $\bM^{(0)}$.
We have
\begin{align*}
    M^{(0)}_{ij}
    &= M_{ij} - \left\langle \frac{r}{n}\bm I_n, \frac{1}{n}\bG_j \bG_j^{\top} \right\rangle - \left\langle \frac{1}{n}\bG_i \bG_i^{\top}, \frac{r}{n}\bm I_n \right\rangle + \left\langle \frac{r}{n}\bm I_n, \frac{r}{n}\bm I_n \right\rangle \\
    &= M^{(1)}_{ij} - \frac{r}{n^2}\|\bG_i\|_F^2 - \frac{r}{n^2} \|\bG_j\|_F^2 + \frac{r^2}{n} \\
    &= M^{(1)}_{ij} - \frac{r^2}{n} - \frac{r}{n^2}(\|\bG_i\|_F^2 - \EE \|\bG_i\|_F^2)- \frac{r}{n^2} (\|\bG_j\|_F^2 - \EE \|\bG_j\|_F^2) \\
    &= M^{(1)}_{ij} - \frac{r^2}{n}\left(1 + \frac{1}{rn}(\|\bG_i\|_F^2 + \EE \|\bG_i\|_F^2) + \frac{1}{rn} (\|\bG_j\|_F^2 - \EE \|\bG_j\|_F^2)\right)
\end{align*}
In matrix terms, we may write
\begin{equation}
    \bM^{(0)} = \bM - \frac{r^2}{n}\left(\one\one^{\top} + \by \one^{\top} + \one \by^{\top}\right) = \bM^{(1)} - \frac{r^2}{n}(\one + \by)(\one + \by)^{\top} + \frac{r^2}{n}\by\by^{\top},
\end{equation}
where $y_i \colonequals \frac{1}{rn}(\|\bG_i\|_F^2 - \EE \|\bG_i\|_F^2)$.
Therefore we have
\begin{align*}
    \lambda_{\min}(\bM)
    &\geq \lambda_{\min}(\bM^{(0)}) - \frac{r^2}{n}\|\by\by^{\top}\| \\
    &\geq \frac{1}{2}r - \frac{r^2}{n}\|\by\|^2 \\
    &= \frac{1}{2}r - \frac{1}{n^3}\sum_{i = 1}^n (\|\bG_i\|_F^2 - \EE \|\bG_i\|_F^2)^2
\end{align*}
The second term is non-negative, and its expectation is, by a routine calculation, equal to $2\frac{r}{n}$.
Thus, with high probability as $n \to \infty$ by Markov's inequality the second term is at most $\frac{1}{4}r$, on which event we have $\lambda_{\min}(\bM) \geq \frac{1}{4}r$.
The result then follows from Lemma~\ref{lem:spectral-lb}.
\end{proof}

\section{Matrix Hyperbolic Cosine Algorithm}
\label{sec:mhc-proofs}

\subsection{Linear Algebra Preliminaries}
\label{sec:linalg}

We review some properties of transcendental matrix functions.
\begin{definition}[Matrix exponential] \label{def:matrix-exponential}
    For $\bM \in \RR^{n\times n}_{\sym}$, the \emph{exponential} of $\bM$, denoted by $e^{\bM}$ or $\exp(\bM)$, is defined as
    \begin{equation} e^{\bM} \colonequals \sum_{k=0}^\infty \frac{1}{k!}\bM^k.\end{equation}
\end{definition}

\begin{proposition}
    Suppose $\bM \in \RR^{n\times n}_{\sym}$ has eigendecomposition $\bM = \bU \bD \bU^\top$, with $\bD$ diagonal. Then $e^{\bM} = \bU e^{\bD} \bU^\top$, where $e^{\bD}$ is diagonal with $(e^{\bD})_{ii} = e^{D_{ii}}$.
\end{proposition}

\begin{definition}[Matrix hyperbolic functions] \label{def:matrix-hyperbolic}
    For $\bM \in \RR^{n\times n}_{\sym}$, we define its \emph{hyperbolic sine} and \emph{hyperbolic cosine} as
    \begin{align}
        \sinh(\bM) &\colonequals \frac{e^{\bM} - e^{-\bM}}{2}, \\
        \cosh(\bM) &\colonequals \frac{e^{\bM} + e^{-\bM}}{2}.
    \end{align}
\end{definition}
It is simple to verify that many properties for scalar hyperbolic functions also apply to the matrix-valued counterparts. We will use these properties without proof when they are self-evident.

The following classical inequality allows us to bound the potential function we work with under addition, and addresses the main property of scalar transcendental functions that is \emph{not} inherited by their matrix versions, that $e^{x + y} = e^x e^y$.
\begin{proposition}[Golden-Thompson inequality {\cite[Equation (IX.19)]{Bhatia-2013-MatrixAnalysis}}] \label{prop:golden-thompson}
    For $\bX, \bY \in \RR^{n \times n}_{\sym}$, $\Tr e^{\bX + \bY} \leq \Tr e^{\bX} e^{\bY}$.
\end{proposition}
\begin{corollary}
    \label{cor:cosh-subadditivity}
    For $\bX, \bY \in \RR_{\sym}^{n \times n}$,
    \begin{equation} \Tr\cosh(\bX + \bY) \leq \Tr\cosh(\bX)\cosh(\bY) + \Tr\sinh(\bX)\sinh(\bY). \end{equation}
\end{corollary}
\noindent
We note that, when matrices are replaced by scalars, then the above inequality holds as an equality and is the standard ``sum rule'' for the hyperbolic cosine.
\begin{proof}
    We expand and apply Proposition~\ref{prop:golden-thompson}:
    \begin{align*}
        \Tr\cosh(\bX + \bY)
        &= \frac{1}{2}\Tr e^{\bX + \bY} + \Tr \frac{1}{2}e^{-\bX - \bY} \\
        &\leq \frac{1}{2}\Tr e^{\bX} e^{\bY} + \frac{1}{2}\Tr e^{-\bX} e^{-\bY} \\
        &= \frac{1}{4}\big( \Tr(e^{\bX} + e^{-\bX})(e^{\bY} + e^{-\bY}) + \Tr(e^{\bX} - e^{-\bX})(e^{\bY} - e^{-\bY}) \big) \\
        &= \Tr\cosh(\bX)\cosh(\bY) + \Tr\sinh(\bX)\sinh(\bY), \numberthis
    \end{align*}
    as claimed.
\end{proof}

\subsection{Khintchine-Type Anti-Concentration Inequalities}

We next provide tools for working with a useful sufficient condition for the MACI property used in our main results.
\begin{definition}
    \label{def:kaci}
    Let $\mu$ be a probability measure on $\RR^n$.
    We say that $\mu$ satisfies a \emph{Khintchine anti-concentration inequality} with \emph{constant} $\eta > 0$ if for all $\bx \in \RR^n$ we have
    \begin{equation}
        \Ex_{\ba \sim \mu}|\langle \bx, \ba \rangle| \geq \eta \|\bx\|_2.
    \end{equation}
    When this is the case, we write that $\mu$ is $\KACI(\eta)$.
\end{definition}

\begin{proposition}
    \label{prop:kaci-maci}
    If $\mu$ a probability measure on $\RR^{n(n + 1)/2}$ is $\KACI(\eta)$, then $\symmat(\mu)$, i.e., the law of $\symmat(\ba)$ for $\ba \in \mu$, is $\MACI(\eta)$.
\end{proposition}
\begin{proof}
    This follows immediately from the definitions along with the inequality $\|\bX\|_F \geq \frac{1}{\sqrt{n}}\|\bX\|_*$ for any $\bX \in \RR^{n \times n}_{\sym}$.
\end{proof}

\begin{proposition}
    \label{prop:kaci-trunc}
    Suppose $\mu$ is a probability measure on $\RR^n$ that is $\KACI(\eta)$ and
    whose covariance matrix has operator norm at most $C > 0$.
    Let $E$ be an event such that $\eta - \sqrt{C (1 - \mu(E))} > 0$.
    Let $\mu^{\prime}$ be either $\mu$ conditioned on $E$, or the measure from which we sample by sampling $\ba \sim \mu$ and observing $\ba$ if $E$ occurs and $\bm 0$ otherwise (a truncation of $\mu$ to the event $E$).
    Then, $\mu^{\prime}$ is $\KACI(\eta - \sqrt{C(1 - \mu(E))})$.
\end{proposition}
\begin{proof}
    Suppose $\mu^{\prime}$ is the truncation of $\mu$ to the event $E$.
    We compute
    \begin{align*}
        \Ex_{\ba \sim \mu^{\prime}}|\langle \bx, \ba \rangle|
        &= \Ex_{\ba \sim \mu}|\langle \bx, \ba \rangle|\, \One\{E\} \\
        &= \Ex_{\ba \sim \mu}|\langle \bx, \ba \rangle| - \Ex_{\ba \sim \mu}|\langle \bx, \ba \rangle|\, \One\{E^c\}
        \intertext{and now bound the first term by the $\KACI(\eta)$ assumption and the second by the Cauchy-Schwarz inequality}
        &\geq \eta \|\bx\|_2 - \sqrt{\mu(E^c)\Ex_{\ba \sim \mu}\langle \bx, \ba \rangle^2 }
        \intertext{and bound the remaining expectation as $\bx^{\top} \Cov(\mu) \bx \leq C\|\bx\|_2^2$, obtaining}
        &\geq \left(\eta - \sqrt{C\, \mu(E^c)}\right)\|\bx\|_2, \numberthis
    \end{align*}
    completing the proof.
    The same holds if $\mu^{\prime}$ is $\mu$ conditioned on $E$ instead; indeed, we have
    \begin{equation}
    \Ex_{\ba \sim \mu}\big[|\langle \bx, \ba \rangle| \,\big|\, E\big] = \frac{\EE_{\ba \sim \mu} |\langle \bx, \ba \rangle| \,\One\{E\}}{\mu(E)} \geq \Ex_{\ba \sim \mu} |\langle \bx, \ba \rangle| \,\One\{E\},
    \end{equation}
    and we may follow the same argument from here.
\end{proof}

\subsection{General Analysis: Proof of Theorem~\ref{thm:mhc}}

The following preliminary result specific to our setting is simple to verify.
\begin{proposition}
    \label{prop:cosh-minus-I}
    Let $\bX \in \RR^{n \times n}_{\sym}$ and $\widetilde{\bX} \colonequals \cosh(\bX) - \bm I$.
    Then, the following hold.
    \begin{enumerate}
        \item $\rank(\widetilde{\bX}) = \rank(\bX)$.
        \item $\widetilde{\bX} \succeq \bm 0$.
        \item $\|\widetilde{\bX}\| = \cosh(\|\bX\|) - 1$.
    \end{enumerate}
\end{proposition}

\begin{proof}[Proof of Theorem~\ref{thm:mhc}]
Denote by $\phi(\bM) \colonequals \Tr\cosh(\alpha \bM)$ the potential function in Algorithm~\ref{alg:mhc}, where $\alpha > 0$ is a parameter of the algorithm that we will choose at the end.
Following the proof strategy of \cite{BS-2020-OnlineBalancingRandom,JKS-2019-OnlineDiscrepancyStochasticEnvy,BJSS-2020-OnlineVectorBalancingDiscrepancy,BJMSS-2021-OnlineDiscrepancyStochastic}, we will first prove the following claims.

\vspace{1em}

\noindent
\textbf{Claim 1:} If $\phi(\bM_{t - 1}) \leq 2n$, then
    \begin{equation}
        \EE_{\bA_t}[\phi(\bM_t)] - \phi(\bM_{t - 1}) \lesssim_{\eta, \theta} \frac{1}{n}.
    \end{equation}

\noindent
\textbf{Claim 2:} If $\phi(\bM_{t - 1}) \geq 2n$, then
    \begin{equation}
        \EE_{\bA_t}[\phi(\bM_t)] - \phi(\bM_{t - 1}) \leq 0.
    \end{equation}

Let us write
\begin{equation}
    \bM_t \colonequals \sum_{i = 1}^t x_i\bA_i.
\end{equation}
Recall that Algorithm~\ref{alg:mhc} chooses
\begin{equation}
    x_t = \argmin_{x \in \{\pm 1\}} \left\{ \phi(\bM_{t - 1} +x\bA_t)\right\} = \argmin_{x \in \{\pm 1\}} \left\{ \phi(\bM_{t - 1} +x\bA_t) - \phi(\bM_{t - 1})\right\}.
\end{equation}
Using this, we analyze the increase in $\phi(\bM_t)$ over one time step as
\begin{align*}
    &\hspace{-0.5cm}\phi(\bM_t) - \phi(\bM_{t - 1}) \\
   &= \phi(\bM_{t - 1} +x_t\bA_t) - \phi(\bM_{t - 1}) \\
   &= \Tr \cosh(\alpha \bM_{t - 1} + \alpha x_t \bA_t) - \Tr\cosh(\alpha \bM_{t - 1})
   \intertext{and by Corollary~\ref{cor:cosh-subadditivity}}
   &\leq  \Tr \cosh(\alpha \bM_{t - 1}) \cosh(\alpha x_t \bA_t) - \Tr\cosh(\alpha \bM_{t - 1}) + \Tr \sinh(\alpha \bM_{t - 1}) \sinh(\alpha x_t \bA_t) \\
   &= \Tr \cosh(\alpha \bM_{t - 1}) (\cosh(\alpha x_t \bA_t) - \bm I) + \Tr \sinh(\alpha \bM_{t - 1}) \sinh(\alpha x_t \bA_t) \\
   &= \langle \cosh(\alpha \bM_{t - 1}), \cosh(\alpha x \bA_t) - \bm I \rangle + \langle \sinh(\alpha \bM_{t - 1}), \sinh(\alpha x_t \bA_t) \rangle
   \intertext{where, since $x_t \in \{\pm 1\}$, $\cosh(-\bY) = \cosh(\bY)$, and $\sinh(-\bY) = -\sinh(\bY)$, the algorithm's choice of $x_t$ will achieve}
   &= \langle \cosh(\alpha \bM_{t - 1}), \cosh(\alpha \bA_t) - \bm I \rangle - |\langle \sinh(\alpha \bM_{t - 1}), \sinh(\alpha \bA_t) \rangle|. \label{eq:cosh-sinh-decomp} \numberthis
\end{align*}

Since $\cosh(\bX) - \bm I \succeq \bm 0$ for any $\bX$ (Proposition~\ref{prop:cosh-minus-I}), we have for the first term of \eqref{eq:cosh-sinh-decomp}
\begin{align*}
    \Ex_{\bA_t} \langle \cosh(\alpha \bM_{t - 1}), \cosh(\alpha \bA_t) - \bm I \rangle &= \langle \cosh(\alpha \bM_{t - 1}), \EE[\cosh(\alpha \bA_t) - \bm I] \rangle \\
    &\leq \langle \cosh(\alpha \bM_{t - 1}), \EE[\|\cosh(\alpha \bA_t) - \bm I\| \bP_{\row(\bA_t)}] \rangle \\
    &\leq (\cosh(\alpha) - 1) \langle \cosh(\alpha \bM_{t - 1}), \EE \bP_{\row(\bA_t)} \rangle \\
    &\leq (\cosh(\alpha) - 1) \|\EE \bP_{\row(\bA_t)} \| \, \Tr \cosh(\alpha \bM_{t - 1}) \\
    &= (\cosh(\alpha) - 1) \theta \frac{r}{n} \, \phi(\bM_{t - 1})
\end{align*}
For the bound of Claim 1, we simply discard the second term and use
\begin{align*}
    \Ex_{\bA_t}[\phi(\bM_t)] - \phi(\bM_{t - 1})
    &\leq \langle \cosh(\alpha \bM_{t - 1}), \cosh(\alpha \bA_t) - \bm I \rangle \\
    &\leq (\cosh(\alpha) - 1) \theta \frac{r}{n} \, \phi(\bM_{t - 1}) \\
    &\leq 2(\cosh(\alpha) - 1) \theta r \\
    &\leq 2\theta \alpha^2 r \\
    &\lesssim \frac{1}{n}, \numberthis
\end{align*}
so long as we choose $\alpha \lesssim 1 / \sqrt{rn}$.

For the bound of Claim 2, we must also analyze the second term of \eqref{eq:cosh-sinh-decomp}, for which we may first bound
\begin{align*}
    &|\langle \sinh(\alpha \bM_{t - 1}), \sinh(\alpha \bA_t) \rangle| \\
    &\hspace{2cm} \geq \alpha |\langle \sinh(\alpha \bM_{t - 1}), \bA_t \rangle| - |\langle \sinh(\alpha \bM_{t - 1}), \sinh(\alpha \bA_t) - \alpha \bA_t \rangle|. \numberthis \label{eq:sinh-bound}
\end{align*}
So long as we take $\alpha \leq 1$ at the end, we will have for the second term in \eqref{eq:sinh-bound} that
\begin{align*}
    \Ex_{\bA_t} |\langle \sinh(\alpha \bM_{t - 1}), \sinh(\alpha \bA_t) - \alpha \bA_t \rangle|
    &\leq \Ex_{\bA_t} \langle |\sinh(\alpha \bM_{t - 1})|, \alpha^3 \bP_{\row(\bA_t)} \rangle \\
    &= \langle |\sinh(\alpha \bM_{t - 1})|, \alpha^3 \EE \bP_{\row(\bA_t)} \rangle \\
    &\leq \alpha^3 \|\EE \bP_{\row(\bA_t)}\| \, \|\sinh(\alpha \bM_{t - 1})\|_*
    \intertext{where we have used the facts that $\|\bA_t\| \leq 1$ and $\sinh(x) - x \leq x^3$ for all $x \in [0, 1]$, and now also using that $|\sinh(x)| \leq \cosh(x)$ we may finish}
    &\leq \alpha^3 \theta \frac{r}{n} \, \phi(\bM_{t - 1}). \numberthis
\end{align*}
Finally, for the first term in \eqref{eq:sinh-bound}, by our assumption of $\MACI(\eta)$, we have\footnote{The computations to come are direct matrix analogs of ones appearing in \cite{BS-2020-OnlineBalancingRandom}.}
\begin{align*}
    \Ex_{\bA_t}|\langle \sinh(\alpha \bM_{t - 1}), \bA_t \rangle|
    &\geq \eta \sqrt{\frac{r}{n^3}} \|\sinh(\alpha \bM_{t - 1})\|_* \\
    &= \eta \sqrt{\frac{r}{n^3}} \sum_{i = 1}^n |\sinh(\alpha \lambda_i(\bM_{t - 1}))| \\
    &\geq \eta \sqrt{\frac{r}{n^3}} \sum_{i = 1}^n (\cosh(\alpha \lambda_i(\bM_{t - 1})) - 1) \\
    &= \eta \sqrt{\frac{r}{n^3}}(\phi(\bM_{t - 1}) - n)
    \intertext{and since we assume $\phi(\bM_{t - 1}) \geq 2n$, we have}
    &\geq \frac{\eta}{2}\sqrt{\frac{r}{n^3}} \phi(\bM_{t - 1}). \numberthis
\end{align*}

Putting together all of these bounds, we find, again supposing we choose $\alpha \leq 1$,
\begin{align*}
    \Ex_{\bA_t} \phi(\bM_t) - \phi(\bM_{t - 1})
    &\leq \left(\theta (\cosh(\alpha) - 1) r + \theta\alpha^3 r - \frac{\eta}{\sqrt{2}}\alpha \sqrt{\frac{r}{n}} \right)\, \frac{1}{n} \phi(\bM_{t - 1}) \\
    &\leq \left(2 \theta \alpha^2 r - \frac{\eta}{\sqrt{2}}\alpha \sqrt{\frac{r}{n}} \right)\, \frac{1}{n} \phi(\bM_{t - 1})
    \intertext{and the choice of $\alpha$ given in the Theorem minimizes this, giving}
    &\leq -\frac{1}{16} \frac{\eta^2}{\theta} \frac{1}{n^2} \phi(\bM_{t - 1}) \\
    &\leq 0. \numberthis
\end{align*}

Combining these results, for any value of $\phi(\bM_{t - 1})$ we have
\begin{equation}
    \Ex_{\bA_t}[\phi(\bM_t)] - \phi(\bM_{t - 1}) \lesssim \frac{1}{n}.
\end{equation}
Thus, taking the expectation over $\bA_1, \dots, \bA_{t - 1}$ as well, we find
\begin{equation}
    \EE[\phi(\bM_t) - \phi(\bM_{t - 1})] \leq \frac{1}{n}.
\end{equation}

Let us write
\begin{equation}
    \widetilde{\phi}(\bM) \colonequals \phi(\bM) - n = \Tr(\cosh(\bM) - \bm I) \geq 0.
\end{equation}
Since $\phi(\bM_0) = \phi(\bm 0) = \Tr\cosh(\bm 0) = n$, we have $\widetilde{\phi}(\bM_0) = 0$, and thus for each $t$
\begin{equation}
    \EE \widetilde{\phi}(\bM_t) \lesssim \frac{t}{n} \leq \frac{T}{n}.
\end{equation}
By Markov's inequality,
\begin{equation}
    \PP\left[\widetilde{\phi}(\bM_t) \geq \frac{T^3}{n} \right] \lesssim \frac{1}{T^2},
\end{equation}
and so, taking a union bound, with high probability $\widetilde{\phi}(\bM_t) \leq T^3 / n \leq T^3$ for all $1 \leq t \leq T$.
On the other hand, using Proposition~\ref{prop:cosh-minus-I}, for any $\bM$ we have
\begin{equation}
    \exp(\|\bM\|) - 1 \leq \cosh(\|\bM\|) - 1 = \|\cosh(\bM) - \bm I\| \leq \Tr(\cosh(\bM) - \bm I).
\end{equation}
Therefore, on the above event, we have for all $1 \leq t \leq T$ that
\begin{equation}
    \|\bM_t\| \leq \frac{1}{\alpha} \log (1 + \widetilde{\phi}(\bM_t)) \lesssim_{\eta, \theta} \sqrt{rn} \, \log\left(1 + \frac{T^3}{n}\right),
\end{equation}
completing the proof.
\end{proof}

\begin{rem}
    \label{rem:mhc-improvements}
    As our proof shows, the term $\log(T)$ in Theorem~\ref{thm:mhc} may be replaced by $\log(1 + T^3 / n)$.
    More generally, for any $f(n) = \omega(1)$, $\log(T)$ may be replaced by $\log(1 + T^2 f(n) / n)$, and if we are only interested in $\|\bM_T\|$ rather than $\max_{t = 1}^T \|\bM_t\|$, then this may moreover be replaced by $\log(1 + T f(n) / n)$, since we do not need to take a union bound over $T$ time points.
    When $T \ll n$, then this latter result recovers that the discrepancy scales essentially as $O(T)$, up to an arbitrarily slowly growing function of $n$.
\end{rem}

\subsection{Independent Entries: Proof of Corollary~\ref{cor:mhc-ind}}

We will use the following condition for KACI on vectors with independent coordinates subject to mild regularity conditions.
\begin{lemma}
    \label{lem:kaci-hc}
    Let $\mu$ be a probability measure on $\RR^n$ such that, for $\ba \sim \mu$, we have:
    \begin{enumerate}
    \item The coordinates of $\ba$ are independent,
    \item $\EE [\ba] = \bm 0$,
    \item $\EE [a_i^4] < \infty$ for all $i \in [n]$,
    \item $\EE[a_i^2] \geq C_1^2$ for all $i \in [n]$ for some $C_1 > 0$, and
    \item $\EE[a_i^4] / (\EE[a_i^2])^2 \leq C_2$ for all $i \in [n]$ for some $C_2 > 0$.
    \end{enumerate}
    Then, $\mu$ is $\KACI(\frac{C_1}{512 C_2})$.
\end{lemma}
\begin{proof}
    The result follows from Fact 3.3, Fact 3.5, and Proposition 3.7 of \cite{GOWZ-2010-FunctionsHalfspacesProductDistributions} (which are themselves standard facts about hypercontractivity): those statements imply that the $a_i$ are (in the language of that paper) all $C_2^{-1/4} / 2\sqrt{3}$-hypercontractive, that $\langle \ba, \bx \rangle$ is $C_2^{-1/4}/ 2\sqrt{3}$-hypercontractive as well, and finally that this implies the anti-concentration inequality
    \begin{equation}
        \PP\left[|\langle \ba, \bx \rangle| \geq \frac{1}{2}\left(\EE \langle \ba, \bx \rangle^2\right)^{1/2} \right] \geq \frac{1}{256 C_2}.
    \end{equation}
    Since the coordinates of $\ba$ are independent and centered, we have $\EE \langle \ba, \bx \rangle^2 = \sum_{i = 1}^n x_i^2 \EE a_i^2 \geq C_1^2\|\bx\|_2^2$.
    Combining these observations, we have
    \begin{equation}
        \EE|\langle \ba, \bx \rangle| \geq \PP\left[|\langle \ba, \bx \rangle| \geq \frac{C_1}{2} \|\bx\| \right] \cdot \frac{C_1}{2} \|\bx\| \geq \frac{C_1}{512 C_2} \|\bx\|,
    \end{equation}
    completing the proof.
\end{proof}

\begin{proof}[Proof of Corollary~\ref{cor:mhc-ind}]
We will apply Theorem~\ref{thm:mhc} to the setup specified in the Corollary.
Let us check the three conditions on $\mu_n$.
First, $\mu_n$ is supported on matrices of operator norm at most 1 by construction.
Second, since we have taken rank sequence $r(n) = n$, $\mu_n$ is vacuously 1-unbiased.

It remains to check that $\sqrt{n} \mu$ is $\MACI(\eta)$ for some constant $\eta > 0$ depending only on the constants $C_i$ in the statement.
By Proposition~\ref{prop:kaci-maci}, it suffices to show instead that $\sqrt{n} \symvec(\mu)$ is $\KACI(\eta)$.
That $\sqrt{n}\symvec(\widetilde{\mu})$ is $\KACI(\eta)$ for some $\eta = \eta(C_1, C_3) > 0$ follows immediately from Lemma~\ref{lem:kaci-hc}, since the assumptions on $\widetilde{\mu}$ in the Corollary subsume the assumptions of the Lemma.
To conclude, we apply Proposition~\ref{prop:kaci-trunc}: since $\mu_n$ is $\widetilde{\mu}_n$ conditioned on $\|\bA\| \leq 1$, and the covariance matrix of $\sqrt{n} \widetilde{\mu}_n$ has norm at most $C_2^2$ by assumption, we find that $\mu_n$ is
\begin{equation}
    \KACI\left(\eta - C_2 \sqrt{\widetilde{\mu}_n(\{\|\bA\| > 1\})}\right).
\end{equation}
For sufficiently large $n$, since $\|\bA\| \leq 1$ with high probability, this quantity will be at least, say, $\eta / 2$, so $\mu_n$ for sufficiently large $n$ is $\KACI(\eta / 2)$.

Thus, Theorem~\ref{thm:mhc} applies to the $\mu_n$, and the proof is complete.
\end{proof}

\subsection{Wishart Matrices: Proof of Corollary~\ref{cor:wishart}}

We will use the following condition on KACI for vectors formed by taking tensor powers of Gaussian vectors.
\begin{proposition}
    \label{prop:poly-kaci}
    Let $\mu$ be the law of the vector of degree $d$ monomials in the entries of $\bg \sim \sN(\bm 0, \bm I_n)$.
    Then, $\mu$ is $\KACI(\eta_d)$ for an absolute constant $\eta_d > 0$ (not depending on $n$).
\end{proposition}
\noindent
This follows from standard polynomial anti-concentration results; see, e.g., \cite{Lovett-2010-GaussianPolynomialAntiConcentration}.

\begin{proof}[Proof of Corollary~\ref{cor:wishart}]
Again, we will apply Theorem~\ref{thm:mhc} to the setup in the Corollary.
First, $\mu_n$ is again supported on matrices of operator norm at most 1 by construction.

The unbiasedness condition is now non-trivial to check: by rotational invariance of Gaussian random matrices, we have that $\row(\bA)$ is a uniformly random $r(n)$-dimensional subspace of $\RR^n$.
Thus, $\bP_{\row(\bA)}$ has the same law as $\bU\bU^{\top}$ where $\bU$ consists of the first $r(n)$ columns of a Haar-distributed $n \times n$ orthogonal matrix.
Basic moment computations for the Haar measure (e.g., Lemma 3.3 of \cite{Meckes-2006-Thesis}) then gives
\begin{equation}
    \EE \bP_{\row(\bA)} = \frac{r(n)}{n} \bm I_n,
\end{equation}
so $\mu_n$ is 1-unbiased with the rank sequence $r(n)$.

Finally we must check the MACI condition, which we again do by checking the KACI condition and invoking Proposition~\ref{prop:kaci-maci}.
Standard results on rectangular Gaussian matrices (see, e.g., Section~7.3 of \cite{Vershynin-2018-HDP}) imply that, with high probability,
\begin{equation}
    \|\bG\| \leq 2(\sqrt{r} + \sqrt{n}) \leq 4\sqrt{n}.
\end{equation}
We have
\begin{align*}
    \EE \left|\left\langle \bX, \frac{n}{\sqrt{r}}\frac{1}{\|\bG\|^2}\bG\bG^{\top} \right \rangle\right|
    &\geq \frac{n}{\sqrt{r}}\EE \left|\left\langle \bX, \frac{1}{\|\bG\|^2}\bG\bG^{\top} \right \rangle\right| \One\{\|\bG\| \leq 4\sqrt{n} \} \\
    &\geq \frac{1}{16\sqrt{r}}\EE |\langle \bX, \bG\bG^{\top}\rangle| \One\{\|\bG\| \leq 4\sqrt{n} \}. \numberthis
\end{align*}

Let us write $\bh \in \RR^{rn(rn + 1)/2}$ for the vector containing each degree 2 monomial in the entries of $\bG$.
By Proposition~\ref{prop:poly-kaci}, the law of $\bh$ is $\KACI(\eta)$ for some $\eta > 0$.
By Proposition~\ref{prop:kaci-trunc}, the law of $\bh$ truncated to the high-probability event that $\|\bG\| \leq 4\sqrt{n}$ is $\KACI(\eta / 2)$ for sufficiently large $n$ (noting that the covariance matrix of $\bh$ is diagonal with entries bounded by an absolute constant).
Thus, we have
\begin{align*}
    \EE |\langle \bX, \bG\bG^{\top}\rangle|\One\{\|\bG\| \leq 4\sqrt{n}\}
    &= \EE \left| \sum_{i, j = 1}^n \sum_{k = 1}^r  G_{ik}G_{jk}X_{ij}\right|\One\{\|\bG\| \leq 4\sqrt{n}\} \\
    &\geq \frac{\eta}{2}\sqrt{r}\|\bX\|_F, \numberthis
\end{align*}
and thus $\frac{n}{\sqrt{r}}\mu_n$ is $\KACI(\eta / 32)$.
Thus, Theorem~\ref{thm:mhc} applies, completing the proof.
\end{proof}

\addcontentsline{toc}{section}{Acknowledgments}
\section*{Acknowledgments}

We thank Afonso Bandeira, David Gamarnik, Antoine Maillard, and Daniel Spielman for helpful discussions.
We especially thank Peng Zhang for pointing out an important error in an early version of this paper, leading to a substantial revision of Theorem~\ref{thm:lb} and its application in Corollary~\ref{cor:lb-proj}.
We also thank an anonymous reviewer for suggesting the proof strategy leading to Corollary~\ref{cor:ub-general}.

\addcontentsline{toc}{section}{References}
\bibliographystyle{alpha}
\bibliography{main.bib}

\end{document}